\documentclass[conference]{IEEEtran}
\ifCLASSINFOpdf
\else
\fi

\usepackage{mathrsfs}
\usepackage{amsmath}
\usepackage{algorithm}
\usepackage{url}
\usepackage{algorithmic}
\usepackage{graphicx}
\usepackage{subfigure}

\begin{document}
%
\title{CrowdFusion: A Crowdsourced Approach on \\Data Fusion Refinement}

\author{\IEEEauthorblockN{Yunfan Chen, Lei Chen, Chen Jason Zhang}
\IEEEauthorblockA{Department of Computer Science and Engineering\\
Hong Kong University of Science and Technology
Hong Kong SAR, China\\
\{ychenbx, leichen, czhangad\}@cse.ust.hk}
}


%


\maketitle

\begin{abstract}
Data fusion has played an important role in data mining because high
quality data is required in a lot of applications. As on-line data
may be out-of-date and errors in the data may propagate with copying
and referring between sources, it is hard to achieve satisfying
results with merely applying existing data fusion methods to fuse
Web data. In this paper, we make use of the crowd to achieve high
quality data fusion. We design a framework selecting a set of tasks
to ask crowds in order to improve the confidence of data. Since data
are correlated and crowds may provide incorrect answers, how to
select a proper set of tasks to ask the crowd is a very challenging
problem. In this paper, we design an approximation solution to
address these challenges since we prove that the problem is at
NP-hard. To further improve the efficiency, we design a pruning
strategy and a preprocessing method, which effectively improve the
performance of the approximation solution. Furthermore, we find that
under certain scenarios, we are not interested in all the facts, but
only a specific set of facts. Thus, for these specific scenarios, we
also develop another approximation solution which is much faster
than the general approximation solution. We verify the solutions
with extensive experiments on a real crowdsourcing platform.
\end{abstract}


%
\IEEEpeerreviewmaketitle

\section{Introduction}

Obtaining true information from Web data is pivotal to the success
of applications in making data-driven decisions. However, the Web
data involve high inconsistency and false/out-of-date information,
as stated in \cite{li2012truth}. As a matter of fact, even for a
simple question such as ``Height of Mount Everest'', we get
conflicting answers from search engines, such as $29,002$ feet,
$29,035$ feet and $29,029$ feet. Essentially, our  task is to
identify the true values from the false ones. Such task is typically
referred to as \emph{data fusion} \cite{dong2014data}.

Data fusion is inherently challenging - the amount of information
available on the Web has been growing rapidly and the Web data are
being altered from time to time. Thus data fusion has drawn much
attention from researchers. Existing works followed a general
principle - they perform a weighted aggregation of multiple data
sources based on the estimated source trustworthiness
\cite{li2015survey}, assuming that the data generated from the same
source is equally reliable. Many works have designed methods of
estimating source quality and integrating results
\cite{zhao2012bayesian,galland2010corroborating,rekatsinas2015finding,dong2012less,li2014resolving}
based on such principle. In addition, there are works trying to
model the relationships between sources
\cite{dong2009truth,pochampally2014fusing,sarma2011data}. Those
methods also follow the same principle.

Though there exist many attempts to address the data fusion problem
as discussed above, these machine-based approaches cannot achieve
high accuracy. Therefore, we  would like to leverage the power of
the crowd to refine data fusion result obtained from pure
machine-based methods. Specifically, we build a crowd-machine hybrid
system due to the following reasons. First, machine-only methods can
hardly overcome the different reliability of Web data sources over
different knowledge domains, but crowd can choose questions about
their familiar domains to answer actively. Real-world data sources
may only offer trustworthy data in some specific domains but are
unreliable in others. For example, in our data set used in the
empirical study, there is an information source called
\emph{eCampus.com} \footnote{\url{http://www.ecampus.com/}} offering
author information of 24 books. For textbook, $55\%$ data offered by
\emph{eCampus.com} are consistent with the author list on the book
covers, however the rate of consistency drops to $0$ for
non-textbooks.  By leveraging crowd power, we can refine the result
obtained from other method and generate a better result.


Second, crowd-only system is too expensive to afford. With help of
computer, we can leverage crowd power efficiently to refine the data
fusion result. Real world facts are correlated with each other. For
instance, the following two facts: \textit{A = ``Barack Obama got
married when he was 31 years old''} and \textit{B=``Barack Obama got
married in 1992''} are related - they are both connected by the fact
that \textit{C=``Barack Obama was born in 1961''}. Explicitly, given
that $A$ and $C$ are both correct, one can infer that $B$ is
correct. Probabilistically, $A$ and $B$ are correlated random
variables, and we have $Pr(A|C) = Pr(B|C)$. We depict the
correlation among the facts by using joint distribution of the
facts. With the correlations, we do not need to ask every fact to
crowd. By asking a subset of questions we can benefit the whole fact
set of interest.

In fact, it is nontrivial to infer correlations between facts. If we
adopt a heuristic method to deal with inference relationships such
as the one proposed in \cite{yin2011semi}, we need to design a
specific function for each single domain. In this work we design a
more general method to describe the correlations among facts with a
joint distribution which will be introduced in
Section~\ref{sect:datamodel}. Our description is general enough to
be compatible with any domain-specific heuristic method.

Third, it is more expensive and inefficient to solve the problems above with other human involved methods such as supervised or semi-supervised\cite{yin2011semi} methods compared to crowdsourcing. This problem is especially severe when the data is sparse and ``Ad-hoc" fusion results are desired. Explicitly, as information is generated fast on the Web and easily gets out-of-date, we need constantly expert efforts to ensure the quality of data fusion system output. For instance, if we search whether the number of planets in the solar system is 8 or 9 on the Web, the number of results of ``9 planets'' is twice of those of ``8 planets'', even if this information has been changed for over $17$ years since Pluto was demoted to Dwarf Planet in 1999.
Asking experts to monitor the updates of information and semi-supervise the system all the time is too expensive and hereby impractical. Compared with semi-supervised methods, crowd is much cheaper and much more ``Ad-hoc'' supported.


In this paper, we develop a novel crowdsourcing-based machine-crowd
hybrid system, namely \textit{CrowdFusion}, to improve result of
exiting \emph{data fusion} methods.
Figure~\ref{figure:system_overview} shows structure of our system.
Compatible with the traditional data fusion models, the input of our
system is a set of fact observations and a prior probability
distribution over all possible results, i.e., probability
distribution calculated by existing data fusion models. There are
two outstanding features of CrowdFusion - (1) it does not hold any
strong assumption as the existing data fusion methods do; and (2)
CrowdFusion is able to make use of the inherent correlations among
facts to work efficiently. Since crowd workers are not always able
to give correct answers especially when the tasks are complex
\cite{park2013query,parameswaran2011human,kittur2011crowdforge}, we
take judgment of one fact as our task to get higher accuracy. The
challenge is how we can ask the crowd tasks efficiently to get more
accurate results with a restricted budget -  we formulate this
challenge as an optimization problem. Furthermore, we show the
computational hardness of the problem, and design an approximate
solution. Besides, we extend the problem to a new scenario with a
query given by users as an extension and propose a task selection
method to ask crowds effectively under such condition.

To summarize, we have made the following contributions.
\begin{itemize}
    \item In Section \ref{sect:def},  we formally define crowdsourcing-based data fusion problem, which is the first attempt to address Web data fusion with the help of the crowd.
    \item We design a system called CrowdFusion to handle data fusion problems, prove that the finding the optimal set of crowdsourcing tasks is NP-hard, and consequently propose several approximation solutions in Section \ref{sect:crowdTS}.
    \item We present a natural extension of crowdsourced data fusion together with its solutions in Section \ref{section:extension}.
    \item Finally, we conduct an extensive experimental study on a real crwodsourcing platform to demonstrate the effectiveneess of CrowdFusion. The results are discussed in Section \ref{sect:exp}.
\end{itemize}
   Besides, Section \ref{related} reviews the related works, and Section \ref{sect:conc} concludes this work as well as suggests possible  directions of future enhancement.

\section{Background and Problem Definitions}\label{sect:def}

In this section, we first formally define the data model, then we
present the crowdsourcing model and finally we give the formal
definition of the problem that we are going to address in this work.

\subsection{Data Model}\label{sect:datamodel}
Let $\mathcal{F}$ be a set of facts and $|\mathcal{F}|=n$. Each fact represents a particular aspect of a real-world entity.
In our problem, the facts are not restricted to one particular \emph{domain}, i.e. $f_i$ and $f_j$ can refer to two totally different real-world entities.

A fact $f_i$ is represented as a triple of \{subject, predicate, object\} and its value is either true or false;
for example, \{Mountain Everest, Height, $29,002 ft$\} cannot be uncertain as there is an
exact height of Mountain Everest. Since the facts are all about real world entities, no fact can be uncertain.
It is quite possible that multiple facts with the same subject and the same predicate are true;
e.g. fact \{Barack Obama, Daughter, Malia Obama\} and fact \{Barack Obama, Daughter, Sasha Obama\} are both true facts.

Facts are related with each other based on similarity, conflict, correlation or any other relationships.
Given $n$ facts, we consider each fact as a Bernoulli random variable, and the dependencies among the facts can be naturally depicted as their joint distribution, which has $2^n$ possible outputs.
We represent probabilities of all possible outputs by $P(o_i), i=1,2,...,2^n$. Output Set $\mathcal{O}$ consists of all possible outputs.

One output $o_i$ is a set of true-or-false judgments, i.e. $o_i = \{(f_i, state)|i=1,...,n, state\in \{true, false\}\}$.
Let $O_k$ be a set of outputs having $f_k$ true, that is $O_k = \{o_i|(f_k, true)\in o_i\}$.

\begin{table}
\centering
\caption{RUNNING EXAMPLE - Facts with Uncertainty}
\centering
\begin{tabular}{|c|c|c|c|c|} \hline
Fid & Entity & Attribute & Value & $P(f_i)$\\ \hline
$f_1$ & Hong Kong & Continent & Asia & $0.5$\\ \hline
$f_2$ & Hong Kong & Population & $\geq 500,000$ & $0.63$\\ \hline
$f_3$ & Hong Kong & Major Ethnic Group & Chinese & $0.58$\\ \hline
$f_4$ & Hong Kong & Continent & Europe & $0.49$\\
\hline\end{tabular}

\label{table:Facts}
\end{table}

\begin{table}
\centering \caption{RUNNING EXAMPLE - Output Joint Distribution}
\centering
\begin{tabular}{|c|c|c|c|c|c|} \hline
Oid & $f_1$ & $f_2$ & $f_3$ & $f_4$ & $P(o_i)$\\ \hline
$o_1$ & F & F & F & F & $0.03$\\ \hline
$o_2$ & F & F & F & T & $0.06$\\ \hline
$o_3$ & F & F & T & F & $0.07$\\ \hline
$o_4$ & F & F & T & T & $0.04$\\ \hline
$o_5$ & F & T & F & F & $0.09$\\ \hline
$o_6$ & F & T & F & T & $0.01$\\ \hline
$o_7$ & F & T & T & F & $0.11$\\ \hline
$o_8$ & F & T & T & T & $0.09$\\ \hline
$o_9$ & T & F & F & F & $0.04$\\ \hline
$o_{10}$ & T & F & F & T & $0.04$ \\ \hline
$o_{11}$ & T & F & T & F & $0.04$\\ \hline
$o_{12}$ & T & F & T & T & $0.05$\\ \hline
$o_{13}$ & T & T & F & F & $0.06$\\ \hline
$o_{14}$ & T & T & F & T & $0.09$\\ \hline
$o_{15}$ & T & T & T & F & $0.07$\\ \hline
$o_{16}$ & T & T & T & T & $0.11$\\
\hline\end{tabular}
\label{table:Distribute}
\end{table}

Table \ref{table:Facts} illustrates the facts with marginal correctness probabilities
and Table \ref{table:Distribute} shows all the possible outputs with the joint probabilities, i.e., $P(o_i)$.
$P(f_k)$ is the marginal probability of how likely $f_k$ is true. That is,
$P(f_k) = \sum_{o_i\in O_k} P(o_i)$.

\newtheorem{definition}{Definition}
\begin{definition}
Given a fact set $\mathcal{F}$, the estimation of quality of $\mathcal{F}$, denoted by $Q(\mathcal{F})$ as an utility function, is the negative value of Shannon Entropy, that is
\begin{displaymath}
Q(\mathcal{F}) = -H(\mathcal{F}) = \sum_{i=1}^{2^n} P(o_i)\log{P(o_i)},
\end{displaymath}
where $\sum_{i=1}^{2^n} P(o_i) = 1$.
\end{definition}

We are using the same mathematical metric as proposed in
\cite{cheng2008cleaning,mo2013cleaning,zhang2015cleaning}, which is
originally called \emph{PWS-quality}. Shannon entropy quantifies the
randomness of random variables, and the lower randomness the output
joint distribution has, the more confident correctness is. Thus, we
can use this utility function to estimate answers' qualities. By
improving the utility of outputs, the confidence of any query
answers would be improved as well. In section \ref{sect:exp}, we
also demonstrate by experiments that the utility function is a good
estimation for answers' quality.

\subsection{Crowdsourcing Model}
As stated before, we ask the crowd whether each fact is true or false independently to keep relatively high crowd reliability. Even if some works assume that crowdsourcing workers do not make any mistake\cite{parameswaran2011human,wang2012crowder,wang2013leveraging,whang2013question}, we consider a more general model which set the probability of correctness of a worker to no less than $0.5$. We will take crowd result as a sample of Bernoulli distribution.
This is a
classical crowdsourcing error model being widely used by many other works\cite{davidson2013using,guo2012so,liu2012cdas,parameswaran2012crowdscreen,sarma2012finding}.
The accuracy can be estimated by a small set of sample tasks with groundtruth.

\begin{definition}
Given a crowd, the probability that answer given by the crowd is correct is $P_c\in[0.5,1]$. We assume all the tasks
completed by crowds are independent from each other, i.e. whether we get correct answer for task $t_i$ does not
affect whether we can get correct answer for task $t_j$ as long as $i \neq j$. We define entropy of crowd $H(Crowd)$ as
\begin{equation}\label{equ:crowdentropy}
H(Crowd) =- P_c\log{(P_c)} - (1-P_c)\log{(1-P_c)}.
\end{equation}

\end{definition}

\begin{definition}
Given a fact set $\mathcal{F}$ and a task set $\mathcal{T}$,
the crowd answer set $\{Ans_i^{\mathcal{T}}|i=1,...,2^{|\mathcal{T}|}\}$ has $2^{|\mathcal{T}|}$ possibilities, each possible answer set contains true or false judgments of $k$ facts in $\mathcal{T}$. Crowd workers answer tasks with
accuracy $P_c$. If $f_i\in\mathcal{T}$, we say that fact $f_i$ is selected. For each possible answer set $Ans_i^{\mathcal{T}}$, the probability of getting that answer set is
\begin{equation}\label{equ:Ans}
P(Ans_i^{\mathcal{T}}) = \sum_{j=1}^{2^n} P(o_j)P_c^{\#\textrm{Same}}(1-P_c)^{\#\textrm{Diff}},
\end{equation}
where $\#\textrm{Same}$ is the number of the same judgment of the selected facts between output set $o_i$ and answer set $Ans_i^{\mathcal{T}}$.  Similarly, $\#\textrm{Diff}$ is the number of different judgment of the selected facts between output set $o_i$ and answer set $Ans_i^{\mathcal{T}}$. Clearly we have
\begin{displaymath}
\sum_{i=1}^{2^{|\mathcal{T}|}}P(Ans_i^{\mathcal{T}}) = 1.
\end{displaymath}
\end{definition}

Given the running example in Table \ref{table:Distribute},
to get an answer set $a_1$ with all negative judgment and $P_c=0.8$, we need to sum up probabilities that the answer set is obtained from crowd under each different potential output.
For $o_1$, $\#\textrm{Same}=4$ and $\#\textrm{Diff}=0$.
Note that correctness of answers are independent from the task and other answers.
Thus we can get the probability that $a_1$ is obtained and that the ground truth equals to $o_1$ is $0.03\times0.8^4\times 0.2^0 = 0.012$.
By summing all probabilities of each potential outputs together we get $P(a_1)=0.049$.
Similarly we can get probability of all possible answers shown in Table \ref{table:ans_dis}.
\subsection{Formalization}

After defining the data model and the crowdsourcing model, we can
now formally define our system goal of CrowdFusion.

\begin{definition}
Given a fact set $\mathcal{F}$, possible outputs $\mathcal{O}$ with
probability joint distribution and a crowd with accuracy $P_c$, our
goal is to maximize the utility $Q(\mathcal{T})$ by selecting a
size-$k$ set of facts to ask the crowd.
\end{definition}

We take crowdsourcing as a powerful tool to refine data fusion result. As a crowd may be noisy and unreliable, we discuss how to use such an imperfect crowd to improve the utility in the rest part of this paper. 

\section{CrowdFusion}\label{sect:crowdTS}

This section introduces a novel crowdsourced data fusion refinement system, namely \textit{CrowdFusion}. The system architecture is demonstrated in Figure~\ref{figure:system_overview}. Explicitly,
CrowdFusion can be initialized by any existing probability-based data fusion method which is demonstrated in Section \ref{related}, or simply set to uniform distribution. After that, the system executes the data improvement process for multiple rounds. In each round, we select a set of tasks, publish them to a crowd, and then use the crowdsourced answers to improve the data quality. The whole procedure terminates when the budget runs out, and generates the fusion results as output.

\begin{figure}[h]\centering
\includegraphics[width=0.40\textwidth]{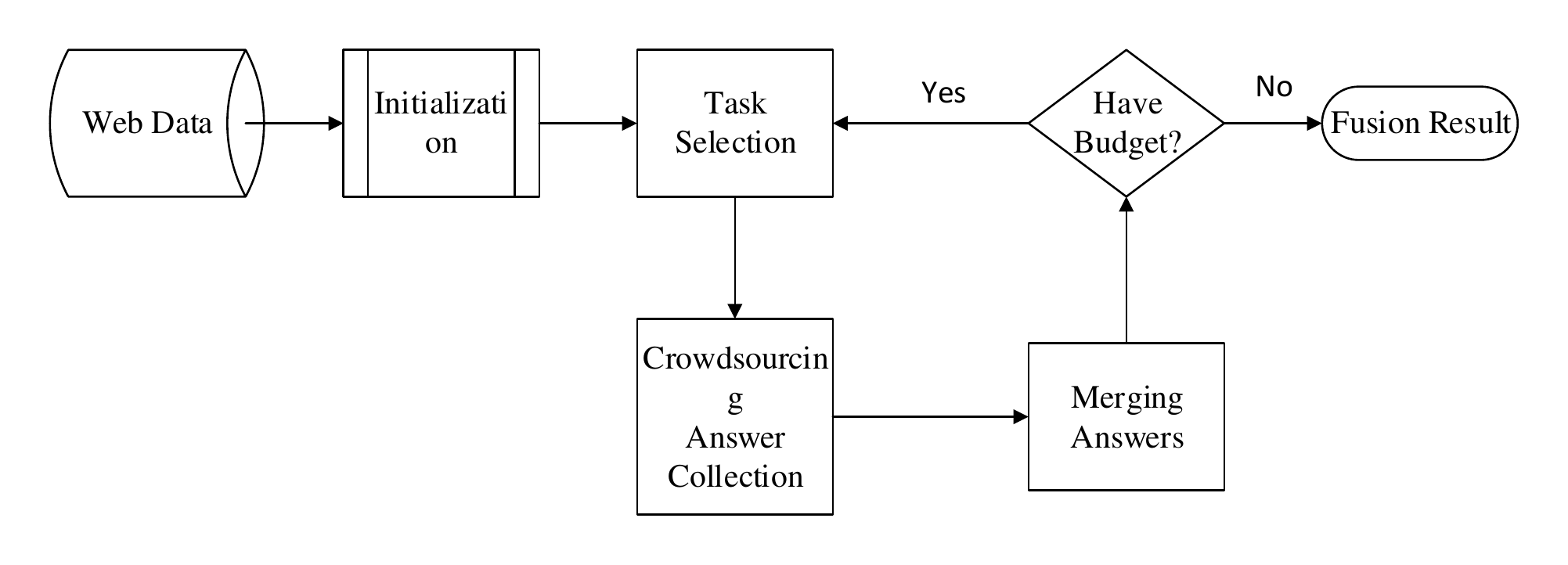}
\centering \caption{System Flowchart}
\label{figure:system_overview}
\end{figure}

In this section, we will discuss how to use the crowd answers to
update the utility value first. Then we introduce the core
optimization problem of this system - task selection. Finally we
discuss solutions to task selection and give theoretical analysis to
those solutions. We call a selection-collection-updating cycle as a
\emph{round} in CrowdFusion.

\subsection{Merging Crowd Answers with Output Set}\label{sect:Merging}
We first look into how a set of obtained answers can affect the output. As facts and crowdsourced answers are both uncertain, merging the crowd answers with the output set can be treated as a posterior output probabilities conditioning on the answers,
and Bayesian theorem can be properly applied on it.

First, $P(o_i)$ is different from $P(Ans_i^{\mathcal{T}})$, even if we ask all the facts to crowds and $o_i$
has the same judgment of facts with $Ans_i^{\mathcal{T}}$.
$P(o_i)$ is the probability that the ground truths of all the facts are $o_i$ but $Ans_i^{\mathcal{T}}$ is the probability that we receive the answers as $Ans_i^{\mathcal{T}}$ which may contains uncertainty involved by the crowd workers.
We can get $Ans_i^{\mathcal{T}}$ by substituting $P_c$ and $\{P(o_i)\}$ in Equation \ref{equ:Ans}.

Then we concentrate on how the crowd answers will affect the confidence of the outputs, or say how to update $P(\mathcal{O})$ after we receive answers $Ans_j^{\mathcal{T}}$ from crowds.
Considering a specific output $o_i$, the probability of $o_i$ to be true is updated to $P(o_i|Ans_j^{\mathcal{T}})$.
Then we can modify the probability of $o_i$ to
\begin{equation}\label{equ:update}
P(o_i|Ans_j^{\mathcal{T}}) = P(o_i)P(Ans_j^{\mathcal{T}}|o_i)/P(Ans_j^{\mathcal{T}}).
\end{equation}

In Equation \ref{equ:update}, $P(Ans_j^{\mathcal{T}})$ is stated by Equation \ref{equ:Ans}. And $P(Ans_j^{\mathcal{T}}|o_i)$ can be calculated by counting the number of the same judgments between $Ans_j^{\mathcal{T}}$ and $o_i$, and we have $\#\textrm{Diff}$ vice versa. Facts unselected are not counted. That is
\begin{displaymath}
P(Ans_i^{\mathcal{T}}|o_i) = P_c^{\#\textrm{Same}}(1-P_c)^{\#\textrm{Diff}}.
\end{displaymath}

Referring to the running example shown in Table \ref{table:Distribute}, we take $\{f_1\}$ as our task set, i.e. we ask ``\emph{Is Hong Kong an Asia city?}'' to the crowd.
And we get an answer of ``yes'' from the crowd with $P_c = 0.8$.
We specify this answer result as event $e$. Then we can update $P(o_1)$:
\begin{eqnarray}
P(e) &=& \sum_{j=1}^{16} P(o_j)P_c^{\#\textrm{Same}}(1-P_c)^{\#\textrm{Diff}} = 0.5 ,\nonumber\\
P(o_1|e) &=& P(o_1)P(e|o_1)/P(e) \nonumber\\
 & = & 0.03 \times (1-0.8) / 0.5 = 0.012 .\nonumber
\end{eqnarray}

The probability that $o_1$ is groundtruth is reduced by the crowd answers as it conflicts with the crowd answer. Meanwhile,
the outputs which state that $f_1$ is true will get higher confidence because of the crowd answers. For example:
\begin{displaymath}
P(o_9|e) = P(o_9)P(e|o_9)/P(e) = 0.04 \times 0.8 / 0.5 = 0.064
\end{displaymath}

\subsection{CrowdFusion Task Selection}\label{sect:Simplify}
Before we discuss how to select the tasks, we need to define the utility improvement after we get answers from the crowd.

\begin{definition}
Given a task set $\mathcal{T}$, the utility of data after asking is:
\begin{displaymath}
Q(\mathcal{F}|\mathcal{T}) = H(\mathcal{T}) - H(\mathcal{F},\mathcal{T}),
\end{displaymath}
where $\mathcal{T}$ stands for task set $\mathcal{T}$ is selected to be asked.
\end{definition}

To maximize $Q(\mathcal{F}|\mathcal{T})$ is as same as to maximize $\Delta Q(\mathcal{F}) = Q(\mathcal{F}|\mathcal{T})-Q(\mathcal{F}) = H(\mathcal{F}|\mathcal{T}) - H(\mathcal{F})$. By properties of entropy, we have:
\begin{displaymath}
\Delta Q(\mathcal{F}) = H(\mathcal{F}) - H(\mathcal{F}|\mathcal{T}) = H(\mathcal{T}) - H(\mathcal{T}|\mathcal{F})
\end{displaymath}

Where $H(\mathcal{T}|\mathcal{F}) = kH(Crowd)$ which is a constant given $k$. Thus, the task is simplified to select
$k$ facts tasks set $\mathcal{T}$ to maximize $H(\mathcal{T})$. And $H(\mathcal{T}) = H(\{Ans_i^{\mathcal{T}}\})$. Formally, we have the following optimization goal:
\begin{equation}\label{equ:goal}
\mathcal{T}_{best} := arg \max_{\mathcal{T}}H(\mathcal{T})
\end{equation}

To further demonstrate our optimization goal, consider a special case that we take $P_c = 1$, $H(\mathcal{T}) = H(\{f_i|f_i \in T\})$. We can just take a subset of size $k$ with highest entropy.
Intuitively, as we know nothing about the crowd, we may choose the best $\mathcal{T}$ with highest $H(\{f_i|f_i \in T\})$ instead of choosing the best $\mathcal{T}$ with highest $H(\mathcal{T})$.
However, this is not true in the general case, i.e. a subset with the highest fact entropy does not indicate
the subset has the highest task entropy. Given the running example in Table \ref{table:Distribute},
if we set $k=2$ and $P_c = 0.8$, that is we want to select two facts to ask whether they are true or false,
the best choices are $f_1$ and $f_4$ based on the calculation results shown in Table \ref{table:Differ}.
 If we trust the crowd and use $H(\{f_i|f_i \in T\})$
instead of $H(\mathcal{T})$, the result, however, will be different.
To be more specific, if $P_c = 1$, the best choices become $f_1$ and $f_2$.

\begin{table}
\centering
\caption{RUNNING EXAMPLE - Difference between Entropy of tasks and of facts}
\centering
\begin{tabular}{|c|c|c|} \hline
$\mathcal{T}$ & $H(\{f_i|f_i \in T\})$ & $H(\mathcal{T})$\\ \hline
$\{f_1,f_2\}$ & \boldmath$1.981$ \unboldmath & $1.993$\\ \hline
$\{f_1,f_3\}$ & $1.949$                & $1.982$\\ \hline
$\{f_1,f_4\}$ & $1.976$                & \boldmath$1.997$\unboldmath\\ \hline
$\{f_2,f_3\}$ & $1.929$                & $1.975$\\ \hline
$\{f_2,f_4\}$ & $1.977$                & $1.993$\\ \hline
$\{f_3,f_4\}$ & $1.948$                & $1.982$\\
\hline
\end{tabular}
\label{table:Differ}
\end{table}

\subsection{Hardness and Challenges}
\newtheorem{theorem}{Theorem}
\begin{theorem}
Given $n$ possible tasks, selecting $k$ of them to reach the highest value of utility function is an NP-hard problem.
\end{theorem}
\begin{proof}
To prove it, it is sufficient to show that the decision version of Task Selection problem is an \emph{NP-Complete} problem.

DTaskSelect: Given a target value $H_t$, is there a selection of $k$ tasks among the $n$ possible facts to reach $H(\mathcal{T}) \geq H_t $.

It is clear that DTaskSelect problem is an \emph{NP} problem. Now we will show that $DTaskSelect\in NPC$ by reducing the Partition Problem which stated as an NPC problem\cite{karp1972reducibility} to it.

In order to make it clear, we define the input probability distribution of DTaskSelect as a set of $\{O_{id},P\}$, where $O_{id}$ is a identifier for the possible outputs in range $[1,2^n]$. Outputs with IDs not given by the input set
have probability of $0$. Meanwhile, each output can be described by a $n$-digits binary number where $i$-th most significant digit is $0$ or $1$ stands for whether fact $f_i$ is true or false respectively.
And we define $O_1, O_2, ..., O_n$ with following expression described format:  $O_i = ((0^x1^x)^i)_2$, where $a^t$ means character $a$ repeat $t$ times and $x=\frac{n}{2^i}$.

The input of PARTITION Problem is a set of numbers $(c_1, c_2, ..., c_s)$. We define $Sum = \sum_{i=1}^s c_s$ and $x_i = \frac{c_i}{Sum}$. By setting $n=2^s$, $o_i=x_i$ for $i=1, 2, ...,n$, $k = 1$, $P_c = 1$ and $H(\mathcal{T}) = 1$, PARTITION is reduced to DTaskSelect.
As we set $P_c = 1$ and $k=1$, assume $T=\{f_i\}$, $H(\mathcal{T})=H(f_i)$. We can calculate $P(f_i)$ by summing up the probability of all the outputs with $f_i=0$ and all the outputs with $f_i=1$ respectively.
By the output format defined above, $f_i$ in outputs $o_1, o_2, ..., o_n$ is exactly a binary number $i$ from the most significant digit to the least.
Thus, with $f_1, f_2, ..., f_n$, we successfully enumerate all the numbers between $0$ to $2^s-1$.
If there exists a partition to get two equal summation sets, we construct a binary number $I=d_1d_2...d_s$ by set $d_i = 0$ iff. $c_i$ and $c_1$ are in the same set.
By choosing $f_I$, we can sum up the exact numbers in the two sets respectively and get the same results of $0.5$. Then we have a selection to reach $H(f_I) = 1$.
On the other side, if there exists a selection $f_I$ satisfy our requirement, where $I=d_1d_2...d_s$, we can do partition to the set to $\{c_i|i \in [1,s],d_i = 0\}$ and $\{c_i|i \in [1,s],d_i = 1\}$.
Easily we can check that the two sets have same summation result.
\end{proof}

\subsection{Approximate Solution}\label{sect:algorithm}

Due to the NP-hardness of the task selection, we can not find the optimal solution within polynomial time unless P=NP.
However, it is known that the conditional entropy is a submodular function \cite{krause2005note}, and the problem of selecting
a k-element subset maximizing a monotone submodular function
can be approximated with a performance guarantee of $(1-1/e)$,
by iteratively selecting the most uncertain variable given the ones
selected so far \cite{nemhauser1978analysis}.
Similarly, we can select $k$ tasks iteratively with a modified greedy algorithm to get a $(1-\frac{1}{e})$-approximate solution.

As we discussed in Section \ref{sect:Simplify}, we want to select $k$ tasks with the maximized entropy of them. To make it clear, we define the gain of evolving one fact $f_j$ into task set $\mathcal{T}$ as
$\rho_j(\mathcal{T}) = H(\mathcal{T}\cup\{f_j\})-H(\mathcal{T})$.

\begin{algorithm}[t]
\caption{Approximation Algorithm for Task Selection}
\label{algo:approx}
\begin{algorithmic}[1]
\REQUIRE
$\mathcal{F}$ , $k>0$;
\ENSURE
Selection $\mathcal{T}\in\mathcal{F}$.

\STATE{Let $\mathcal{T}^0 = \phi$, $\mathcal{F}^0 = \mathcal{F}$, and set $i=1$.}

\REPEAT
    \STATE{Select $t(i)\in \mathcal{F}^{i-1}$ for which $\rho_{t(i)}(\mathcal{T}^{i-1})=\max_{t\in \mathcal{F}^{i-1}}\rho(\mathcal{T}^{i-1})$.}\label{code:select}
    \STATE{Set $\rho_{i-1}=\rho_{t(i)}(\mathcal{T}^{i-1})$.}

    \IF{$\rho_{i-1}\leq 0$}
        \STATE{Break loop with the set $\mathcal{T}^{K^*}$, where $K^* = i-1 < k$}
    \ELSE
        \STATE{Set $\mathcal{T}^i = \mathcal{T}^{i-1} \cup \{t(i)\}$.}
        \STATE{And set $\mathcal{F}^i = \mathcal{F}^{i-1} = \{t(i)\}$.}
    \ENDIF

    \STATE{Set $i\gets i+1$.}
\UNTIL{$i=k+1$(i.e.$K^* = k$).}
\RETURN $\mathcal{T}^{K^*}$
\end{algorithmic}
\end{algorithm}

In Algorithm \ref{algo:approx}, we select $K^*$ facts instead of $k$ facts, where $K^* \leq k$.
In most cases, $K^* = k$. If we cannot get any benefit from asking one more task, we will stop selecting and will cause $K^* < k$.
However, unless $k > n$ or we are certain about every fact remaining to be selected, we will continue asking until reaching $k$ tasks because of the following theorem.

\begin{theorem}\label{theorem:Kstar}
Utility will be improved whenever an uncertain fact exists to be selected to ask.
\end{theorem}
\begin{proof}
Assume in the algorithm we already selected out a task set $\mathcal{T}$ and there is a fact $f_j$ with $0 < P(f_j = true|\mathcal{T}) < 1$.
Based on definition, $\rho_j(\mathcal{T}) = H(\mathcal{T}\cup\{j\})-H(\mathcal{T}) = H(Ans_j^\mathcal{T'}|\mathcal{T})$, where $\mathcal{T'} = \mathcal{T}\cup f_j$.
As each tasks are answered independently, we have $H(Ans_j^\mathcal{T'}|\mathcal{T}) = H(Ans_j^\mathcal{T'})$.
Since given $\mathcal{T}$, we are uncertain about whether fact $f_j$ is true or not, $P(Ans_j^\mathcal{T'} = true)=P_cP(f_j = true|\mathcal{T}) + (1-P_c)P(f_j = false|\mathcal{T}) > 0$.
Then we have $\rho_{i-1}>0$, we will select one more task by the algorithm.
\end{proof}

In the Algorithm \ref{algo:approx}, we use marginal distribution to calculate $\rho$ accordingly. We do not have to calculate the marginal distribution from nothing in each iteration.
Note that in Equation \ref{equ:Ans}, if we extend $\mathcal{T}$ with an additional fact, the only modification is that \emph{\#Same} or \emph{\#Diff} will be increased by $1$.
That is, we can take every fact into our task set first. And in each iteration, we can separate the task set into two subsets according to the value of the fact chosen in this iteration. We double the number of subsets in each iteration and we only need a linear scan to finish this operation.
By summing up probabilities in each set we get a marginal distribution for this iteration and can calculate the entropy accordingly.

\begin{table}
\centering \caption{RUNNING EXAMPLE - Answer Joint Distribution}
\centering
\begin{tabular}{|c|c|c|c|c|c|} \hline
$Ans_i$ & $f_1$ & $f_2$ & $f_3$ & $f_4$ & $P(a_i)$\\ \hline
$a_1$ & F & F & F & F & $0.049$\\ \hline
$a_2$ & F & F & F & T & $0.050$\\ \hline
$a_3$ & F & F & T & F & $0.063$\\ \hline
$a_4$ & F & F & T & T & $0.055$\\ \hline
$a_5$ & F & T & F & F & $0.071$\\ \hline
$a_6$ & F & T & F & T & $0.049$\\ \hline
$a_7$ & F & T & T & F & $0.087$\\ \hline
$a_8$ & F & T & T & T & $0.077$\\ \hline
$a_9$ & T & F & F & F & $0.047$\\ \hline
$a_{10}$ & T & F & F & T & $0.051$ \\ \hline
$a_{11}$ & T & F & T & F & $0.052$\\ \hline
$a_{12}$ & T & F & T & T & $0.056$\\ \hline
$a_{13}$ & T & T & F & F & $0.065$\\ \hline
$a_{14}$ & T & T & F & T & $0.071$\\ \hline
$a_{15}$ & T & T & T & F & $0.073$\\ \hline
$a_{16}$ & T & T & T & T & $0.085$\\
\hline\end{tabular}
\label{table:ans_dis}
\end{table}

Look back at the running example in Table \ref{table:Distribute}.
We can first calculate $\{P(Ans_j^{\mathcal{T}_{all}})\}$ which is the probability distribution of answers if we take every fact as our task set as shown in Table \ref{table:ans_dis} by setting $P_c = 0.8$. In the first round, we calculate marginal distribution of every single fact and we select $f_1$ in this round because entropy of selecting $f_1$, $H(\{Ans_j^{\{f_1\}}\}) = 1$, is the maximum entropy among all selections.
And we now calculate the entropy for the second round. We firstly separate the task set into two subsets by the value of $f_1$. Then we separate the two subsets by every fact remaining to be selected.
Then we can get $H(\{Ans_j^{\{f_1, f_4\}})\} = 1.997$ and $f_4$ is the best choice in this iteration.
Following algorithm \ref{algo:approx}, we get our selection $f_1$ and $f_4$.

\subsection{Pruning}
By considering the highest possible total gain value after selecting a fact in each iteration, we can prune some facts and do not need to consider them any more in the following iterations of the approximation solution.
That is, we need an upper bound for $\rho_{j\cup\mathcal{S}}(\mathcal{T}^i)$, where $\mathcal{S}$ is any possible further task set.
And the final task set is $\mathcal{T}^i\cup f_j \cup \mathcal{S}$.

Define that Task set in $i$-th iteration with an additional task $f_j$ as $\mathcal{T}^i_j$.
By definition, we have
\begin{equation}\label{equ:prune1}
\rho_{j\cup\mathcal{S}}(\mathcal{T}^i) = H(\mathcal{T}\cup\{f_j\}\cup\mathcal{S}) - H(\mathcal{T}).
\end{equation}
If for a fact whose highest total gain is less than any gain we have in this step, we do not need to consider it any more in this task selection algorithm, that is we prune it.

With property of entropy, we have
\begin{eqnarray}\label{equ:prune2}
H(\mathcal{T}\cup\{f_j\}\cup\mathcal{S}) &=& H(\mathcal{T}\cup\{f_j\}) + H(\mathcal{S}) - I(\mathcal{T}\cup\{f_j\}\;\mathcal{S}) \nonumber\\
& \leq & H(\mathcal{T}\cup\{f_j\}) + H(\mathcal{S}) \nonumber\\
& \leq & H(\mathcal{T}\cup\{f_j\}) +\log{(k-|\mathcal{T}|-1)}
\end{eqnarray}

In each iteration, as the base values $H(\mathcal{T})$ are the same, considering total gain is equivalent to considering the possible total entropy $H(\mathcal{T}\cup\{f_j\}\cup\mathcal{S})$.
Since that we always gain more by asking one more task by Theorem \ref{theorem:Kstar}, we cannot apply the pruning strategy for the first fact we tried in each iteration.
After calculation $H(\mathcal{T}\cup\{f_j\})$ for the first time, we use it as an initial value of maximum entropy in this iteration.
The maximum entropy will be updated if we get a higher entropy in following enumeration of facts in this iteration.
For any further calculation, we add the entropy result with $\log{(k-|\mathcal{T}|-1)}$ which is a constant in each iteration and can be computed in constant time.
If the upper bound value is lower than the maximum entropy till now, we do not need to consider the related fact as a task any more.
Not only in this iteration but also any possible further iteration.

\begin{theorem}
Pruning $f_j$ for all following selections is safe if $H(\mathcal{T}\cup\{f_j\})+\log{(k-|\mathcal{T}|-1)}<\max_t{H(\mathcal{T}\cup\{f_t\})}$.
\end{theorem}
\begin{proof}
The correctness of this pruning can be easily addressed by contradiction proving.
Assume that our greedy algorithm returns a task set with a fact $f_j$ which is pruned in iteration $i$.
Because $f_j$ is pruned, there exists a fact $f_k$ in iteration $i$ with $H(\mathcal{T}^i_k) \geq H(\mathcal{T}^i_j)+\log{(k-|\mathcal{T}|-1)}$ and $f_j$ is not in set $\mathcal{T}^{i}$.
If $f_j$ is in final task set, that means in $t$-th iteration we selected $f_j$ where $t>i$.
However, the entropy is not related to task sequence and will increase with any additional task selected.
If we take $\mathcal{T^i}$ and $f_j$ as a subset of final task set, we cannot get entropy greater than $H(\mathcal{T}^i_k)$.
However, the final task set entropy must be no smaller than $H(\mathcal{T}^i_k)$ as we stated before that every task we evolved we can get higher gain.
Here comes the contradiction and the pruning is safe.
\end{proof}

\subsection{Greedy with Preprocessing}
By the Algorithm \ref{algo:approx}, we can achieve result with approximation rate $(1-\frac{1}{e})$.
In each iteration, we need to calculate the entropy for $O(n)$ possible selections.
According to the Equation \ref{equ:goal}, we suppose to select $k$ tasks to achieve highest $H(\mathcal{T}) = H(Ans_i^{\mathcal{T}})$.
To be more specific, the closed-form of computing $H(Ans_i^{\mathcal{T}})$ is:
\begin{eqnarray}
H(Ans_i^{\mathcal{T}}) &= & - \sum_{j=1}^{2^n} [P(o_j)P_c^{\#\textrm{Same}}(1-P_c)^{\#\textrm{Diff}} \nonumber\\
 & & \cdot \log{(P(o_j)P_c^{\#\textrm{Same}}(1-P_c)^{\#\textrm{Diff}})}]\nonumber
\end{eqnarray}

Note that the Algorithm \ref{algo:approx} runs in $O(k\cdot T)$ time, where $T$ is time cost of line \ref{code:select} in the algorithm.
If we calculate it with a brute force method, we will have to calculate the marginal distribution of each choice.
For any of choice, we have $O(2^k)$ items in the marginal distribution and for each item we need $O(k|\mathcal{O}|)$ to calculate the probability as we need time $k$ for counting the number of the same and the different judgments between the output items with the answer.
Thus, we need $O(2^knk|\mathcal{O}|)$ time for a single step which is a huge cost as it will repeat $k$ times and finally cost $O(2^knk^2|\mathcal{O}|)$ in total.

In order to reduce the overhead, we propose a preprocessing in $O(|\mathcal{O}|^2)$ time, which can reduce the total cost of the algorithm to $O(kn\mathcal{O})$.
Even though the preprocessing costs much, it can be finished off-line, and it will be useful when we need to generate multiple task sets.
Further more, the preprocessing has good property and can be solved by parallel computing or the MapReduce framework in constant time with $|\mathcal{O}|^2)$ cores.

The preprocessing is the calculation of Answer Joint Distribution as
shown in the running example of Table \ref{table:ans_dis}. As the
size of the answer joint distribution is also $|\mathcal{O}|$, we
need to scan the output distribution data to calculate each item of
the answer joint distribution which costs $O(|\mathcal{O}|^2)$ in
total. However, as the scan sequence does not affect the result, we
just need to calculate the summation of each cross relationship
between answers and outputs. Each of sub-program is responsible for
one single counting and calculation of
$P_c^{\#\textrm{Same}}(1-P_c)^{\#\textrm{Diff}}$. We sum them up by
the \emph{id}s of the answer the sub-program responsible for which
is easy to be solved by the MapReduce
framework\cite{dean2008mapreduce}, any other parallel computing or
any many core computing.

With the preprocessing result, we can significantly accelerate line \ref{code:select} of the Algorithm \ref{algo:approx}.
In that step, we need to calculate the marginal distribution for each choice and facts already chosen.
Since every marginal probability is a summation of a set of output items, we can separate original distribution $O(k)$ times and calculate by an additional scan totally in $O(k|\mathcal{O}|)$ time which can be reduced to $O(|\mathcal{O}|)$ later.

\begin{algorithm}[t]
\caption{Compute Marginal Distribution}
\label{algo:margin}
\begin{algorithmic}[1]
\REQUIRE
Answers joint distribution , $\mathcal{T}$;
\ENSURE
Marginal Distribution $P(\{Ans^{\mathcal{T}}_i\})$.

\STATE{Initially, answer set has one part as a whole.}

\REPEAT
    \STATE{Select a fact $f_i\in \mathcal{T}$.}
    \STATE{$\mathcal{T} \gets \mathcal{T} - {f_i}$.}
    \STATE{Separate each part of the Answer joint distribution by judgment of $f_i$ into two new parts}
\UNTIL{$\mathcal{T}=\phi$.}
\STATE{Sum up each separated part to gets $\{p_i\}$}.
\RETURN $\{p_i\}$.
\end{algorithmic}
\end{algorithm}

With the Algorithm \ref{algo:margin}, we can simply sum up $-\sum_{p_i} p_i\log{p_i}$ to get $H(\mathcal{T})$.
The correctness of this algorithm is straightforward as we simply separate output items with same judgment of $\mathcal{T}$ into a same part and sum them up.
And this algorithm scans the answers table for each separation, which takes $O(|\mathcal{O}|)$ for each separation.
There are $k$ separations in total as we remove one fact each time from $\mathcal{T}$ until $\mathcal{T}$ is an empty set and $\mathcal{T}$ has $k$ items at most.
Then we can get that the total cost of this algorithm is $O(k|\mathcal{O}|)$ time and the total cost for line \ref{code:select} in the Algorithm \ref{algo:approx} is $O(nk|\mathcal{O}|)$.

Further more, as we select only one fact in each iteration, the total cost for the step can be reduced to $O(n|\mathcal{O}|)$.
Just by storing separation result of last iteration, we only needs one separation and calculation for each choice which costs $O(|\mathcal{O}|)$ time.
That is, right after each calculation of the marginal distribution, we calculate the entropy $H(\mathcal{T})$ and record the separating result if this is the maximum entropy in this iteration.
For the next iteration, we only need to add one task to task set $\mathcal{T}$, which only requires one separation for that.
With such a strategy, we can reduce the cost of this step to $O(n|\mathcal{O}|)$ and the total cost is reduced to $O(nk|\mathcal{O}|)$.

With an input, our system first selects a task set with brute force
or approximation solution and after collecting answers from crowd,
our system updates the output set. Such a select-collect-update
procedure is a round of our system. As long as we have budget, we
run another round to get the final result. The prune strategy and
preprocessing strategy are proposed to accelerate the approximation
solution.

\section{Extension -- Query Based CrowdFusion}
\label{section:extension}
Sometimes users of data fusion system are only interested in some aspects of an entity.
For instance, some users want to see what effectiveness of political policies in different population and major people condition.
Under such a requirement, we need accurate population and people information for each country or region.
However, continent is still a task worth to be asked to the crowd even if we do not care about continent.
This is because there is relationship between continent and major people. Continent information also infers population, e.g. Asia countries tend to have large population than any other continent's countries.
Query based CrowdFusion is designed to solve such a condition.
If we are not interested in all aspects, we can get higher accuracy by asking fewer tasks.

The input of this system remains the same as other aspects may support some of interested facts as described before.
In order to describe the query-based demands, the utility function and the outputs will be modified accordingly.
As shown later, query-based condition also admits submodularity and can be solved by the previous Algorithms \ref{algo:approx}
with simple modifications to achieve the same approximation rate.

\subsection{Data Model}
Note that in the original data model shown in section \ref{sect:datamodel}, each output $o_i$ is a true-or-false judgment for all facts.
As we are only interested in a subset of facts for this new scenario, we do not have to judge all facts.
Define that facts of interest(FOI) set $\mathcal{I}$ which is a subset of original fact set, $\mathcal{I}\subseteq\mathcal{F}$.
One output $o_i$ is modified to a set of true-or-false judgment of FOI, i.e. $o_i = \{(f_i, state)|f_i \in \mathcal{I}, state\in \{true, false\}\}$.
Similarly, we define utility function like $Q(\mathcal{I}) = -H(\mathcal{I})$.

We cannot use original updating methods by replacing output set to the new one,
because not all answers will appear in the output set, i.e. $P(Ans_i^{\mathcal{T}}|o_i) \neq P_c^{\#\textrm{Same}}(1-P_c)^{\#\textrm{Diff}}.$
In order to calculate it, we need to use the original output set to update the probability of all the outputs,
and then calculate the marginal distribution to get the outputs about FOI.
\subsection{Solution}
Obviously, query based CrowdFusion is a general case of CrowdFusion since we can reduce query based CrowdFusion to overall cases by setting $\mathcal{I}=\mathcal{F}$.
Thus, query based CrowdFusion is at least NP-hard and one can hardly find Polynomial time algorithm for it.

Fortunately, selecting a set of tasks to ask in order to maximize query-based utility function can be approximated.
By properties of entropy, we have $Q(\mathcal{I}|\mathcal{T}) = H(\mathcal{T}) - H(\mathcal{I},\mathcal{T})$.
Suppose there are two task sets $\mathcal{T}$ and $\mathcal{T}'$ satisfying $\mathcal{T} \subseteq \mathcal{T}'$.
Straightforwardly we have $H(\mathcal{I},\mathcal{T})\leq H(\mathcal{I},\mathcal{T}')$. And finally we have for all $\mathcal{T} \subseteq \mathcal{T}'$,
\begin{equation}
Q(\mathcal{I}|\mathcal{T}) \geq Q(\mathcal{I}|\mathcal{T}').
\end{equation}
Thus we proved sub-modularity of the query-based utility function. This query based version problem can be solved by algorithm \ref{algo:approx} by simply substituting $\rho_j(I)$ with $\rho_j(I) = Q(\mathcal{I}|T \cup \{j\})-Q(\mathcal{I}|T)$.

Please note that even if we use same algorithm framework as before to solve this problem, we need to calculate marginal distribution before computing the utility improvement $\rho$ which is a huge cost.
As the algorithm and the function property remains the same, with such a method, we get a $(1-\frac{1}{e})$-approximate solution.

\section{EXPERIMENTAL EVALUATION}\label{sect:exp}
\begin{table*}[ht]
	\centering
	\label{table:time}
	\caption{One Round Average Running Times of Five Approaches}
	\begin{tabular}{|l|c|c|c|c|c|} \hline
		$k$&\emph{OPT}&Approx.&Approx.\&Prune & Approx.\&Pre. & Approx.\&Prune\&Pre. \\ \hline
		$1$ & $37.78$ & $32.60$ & $33.44$ & $\mathbf{1.08}$ & $1.15$\\ \hline
		$2$ & $1475.66$ & $94.73$ & $40.94$& $2.10$ & $\mathbf{1.34}$\\ \hline
		$3$ & $75359.26$ & $242.22$ & $56.66$& $3.10$ & $\mathbf{1.40}$\\ \hline
		$4$ & & $598.15$ & $74.93$& $4.08$ & $\mathbf{1.53}$\\ \hline
		$5$ & & $1401.05$ & $74.32$& $4.96$ & $\mathbf{1.54}$\\ \hline
		$6$ & & $3230.22$ & $76.09$& $5.86$ & $\mathbf{1.48}$\\ \hline
		$7$ & & $7005.02$ & $74.35$& $6.67$ & $\mathbf{1.63}$\\ \hline
		$8$ & & $14611.53$ &$ 74.88$& $7.80$ & $\mathbf{1.53}$\\ \hline
		$9$ & & $29476.42$ & $74.87$& $8.37$ & $\mathbf{1.59}$\\ \hline
		$10$ & & $57198.67$ & $74.34$& $9.39$ & $\mathbf{1.82}$\\ \hline
	\end{tabular}
	
	\small Note: Data listed in this table are time costs for each conditions in second.
\end{table*}
We conducted extensive experiments to evaluate our proposals with real-world datasets on gMission \cite{chen2014gmission}, which is a
public crowdsourcing platform. We focus on
investigating three issues. First, we examine the efficiency of our techniques.
Second, we verify the effectiveness of our approaches, by evaluating
the utility and $F_1$-score. Third, we analyze the reasons that cause the inevitable errors existing the fusion results.

\subsection{Experiment Setup}
\textbf{Crowdsourcing Platform:}We conduct our experiments on \textit{gMission} - a real crowdsourcing platform.
Tasks generated by our CrowdFusion system are posted on that gMission and then pushed to crowd workers.
Each task is answered independently by a number of anonymous gMission users, and they share an accuracy rate $P_c$ as described before.

\begin{figure*} [ht]
	\centering 
	\subfigure[$P_c=0.7$, $F_1$-Score]{ 
		\label{fig:bruteQuality:a} 
		\includegraphics[width=0.3\textwidth]{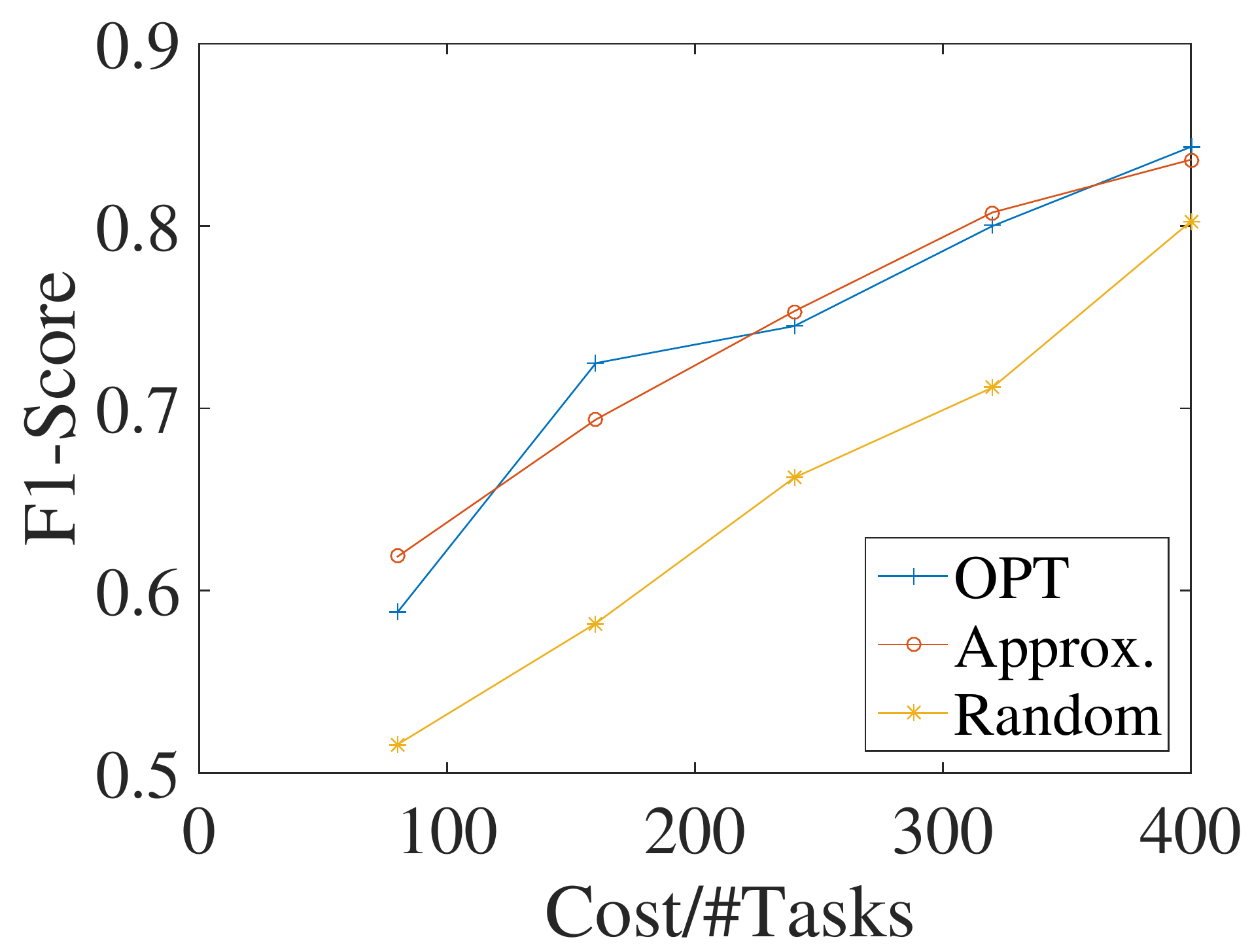}}  
	\subfigure[$P_c=0.8$, $F_1$-Score]{ 
		\label{fig:bruteQuality:b} 
		\includegraphics[width=0.3\textwidth]{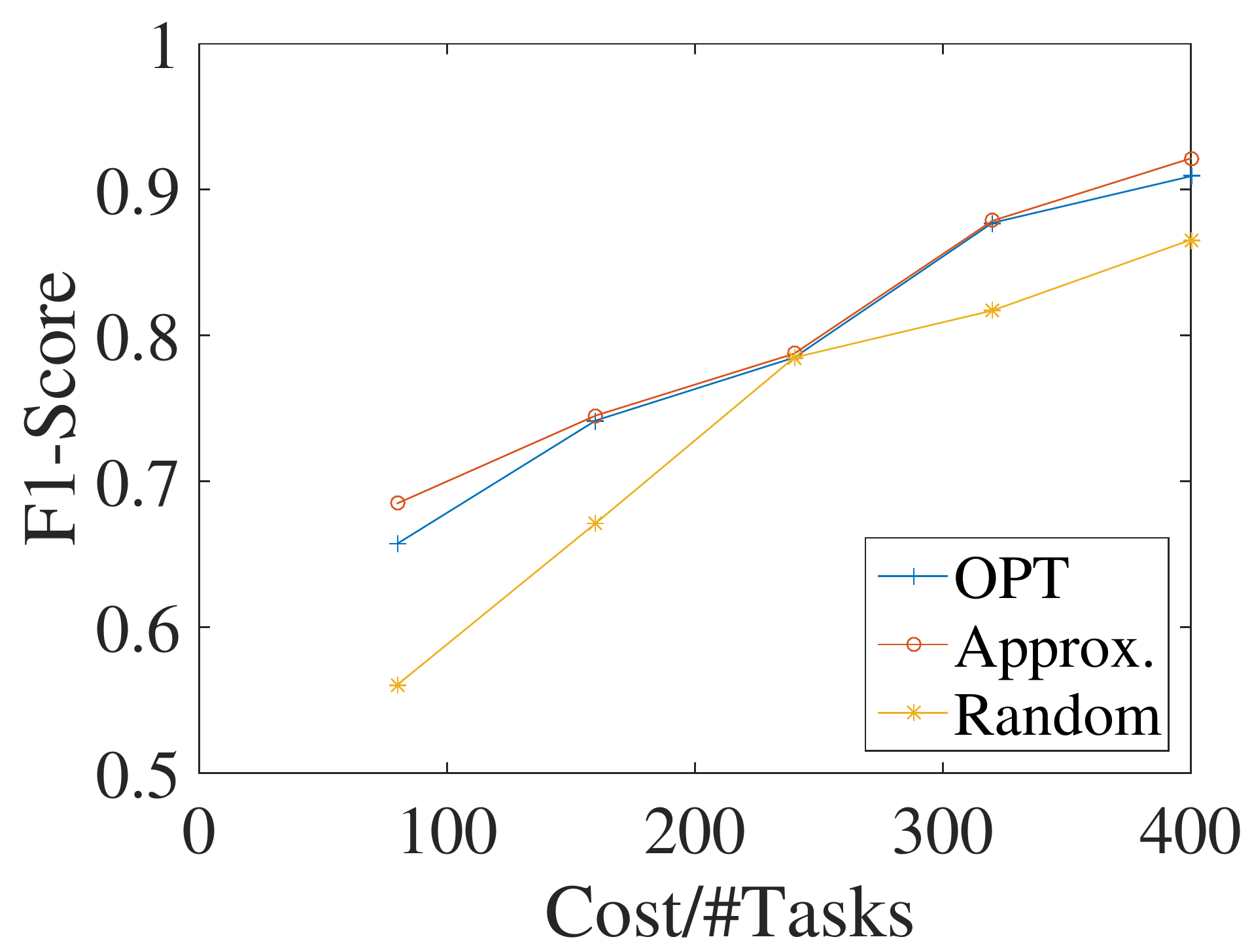}} 
	\subfigure[$P_c=0.9$, $F_1$-Score]{ 
		\label{fig:bruteQuality:c} 
		\includegraphics[width=0.3\textwidth]{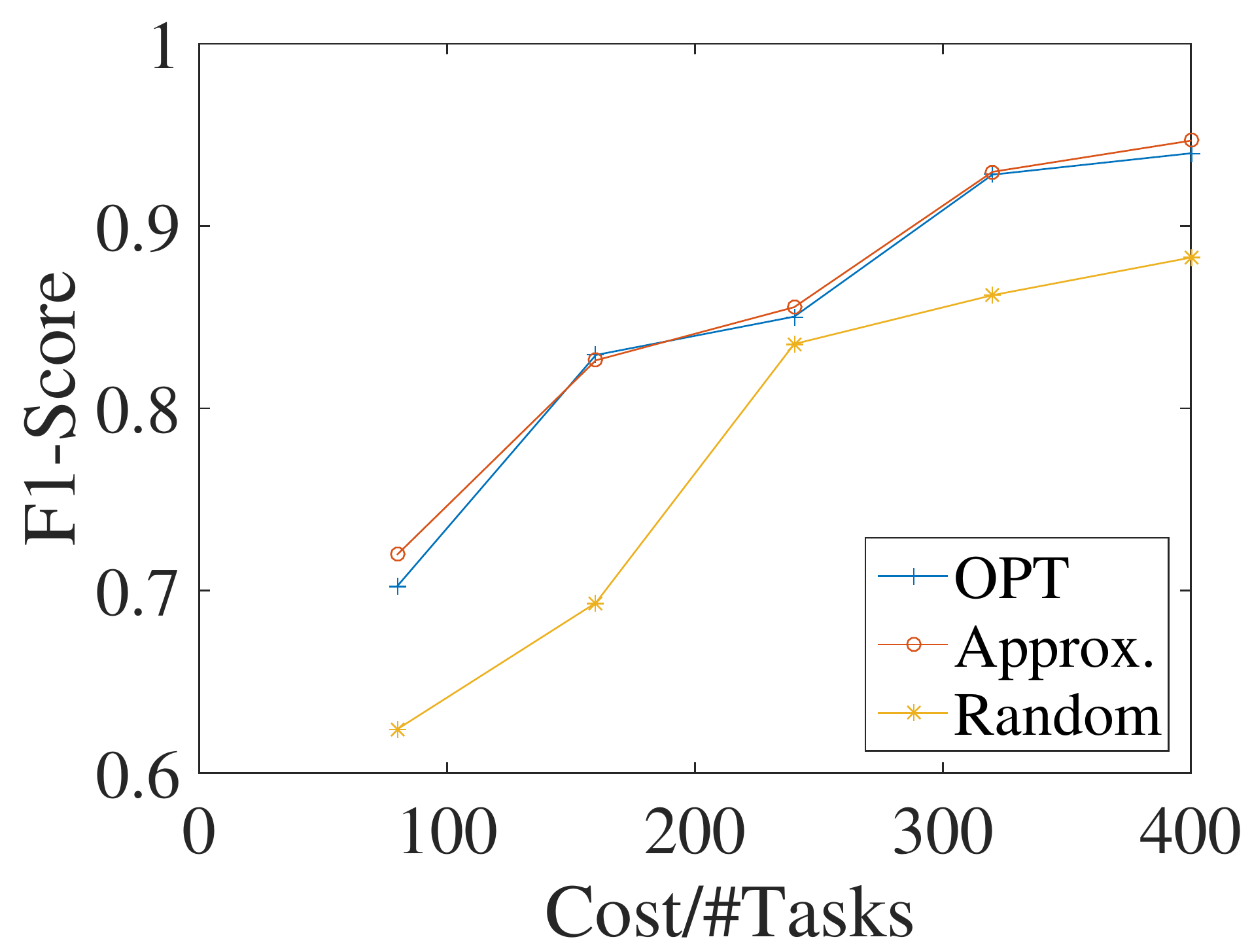}} 
	\subfigure[$P_c=0.7$, Utility]{ 
		\label{fig:bruteQuality:d} 
		\includegraphics[width=0.3\textwidth]{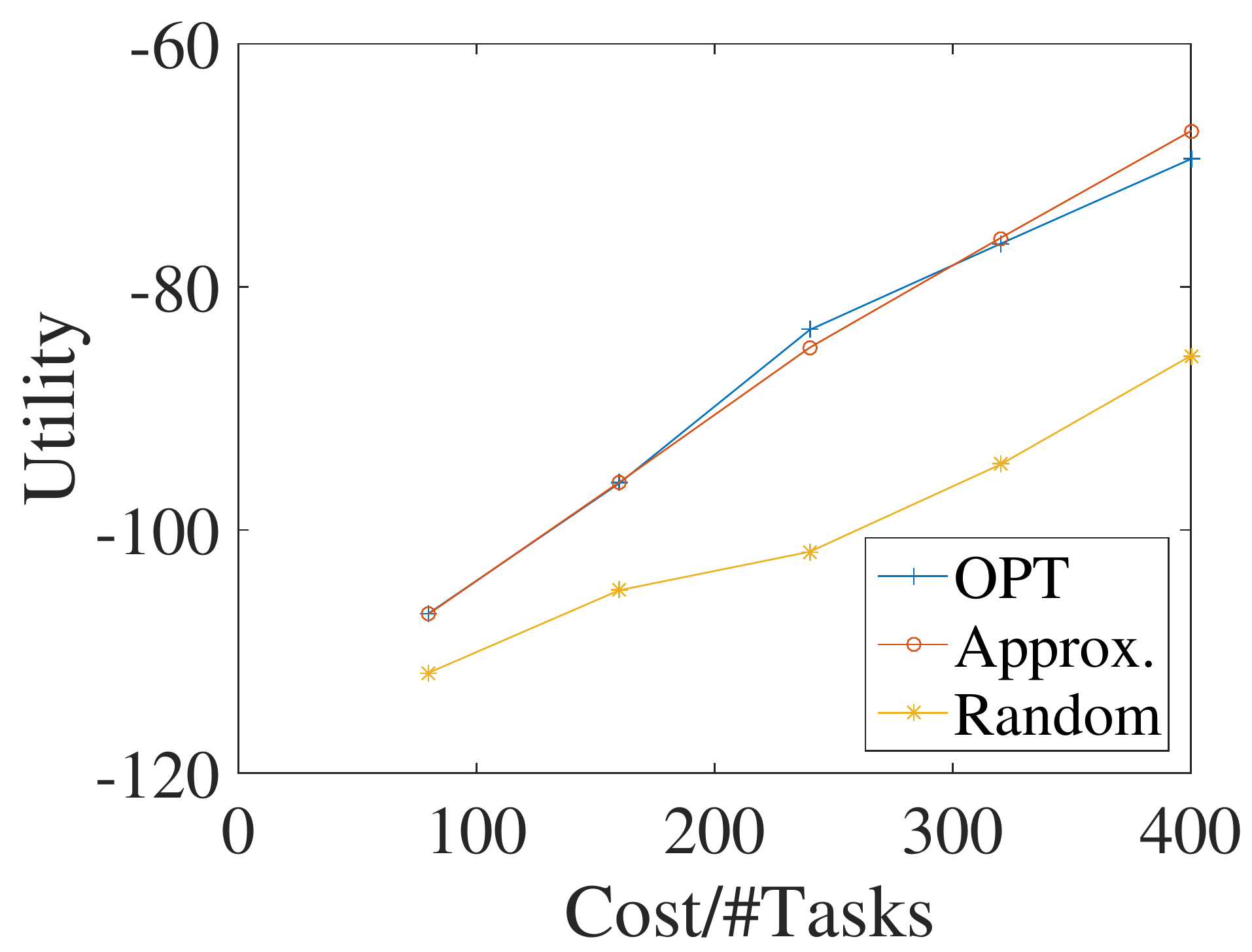}}  
	\subfigure[$P_c=0.8$, Utility]{ 
		\label{fig:bruteQuality:e} 
		\includegraphics[width=0.3\textwidth]{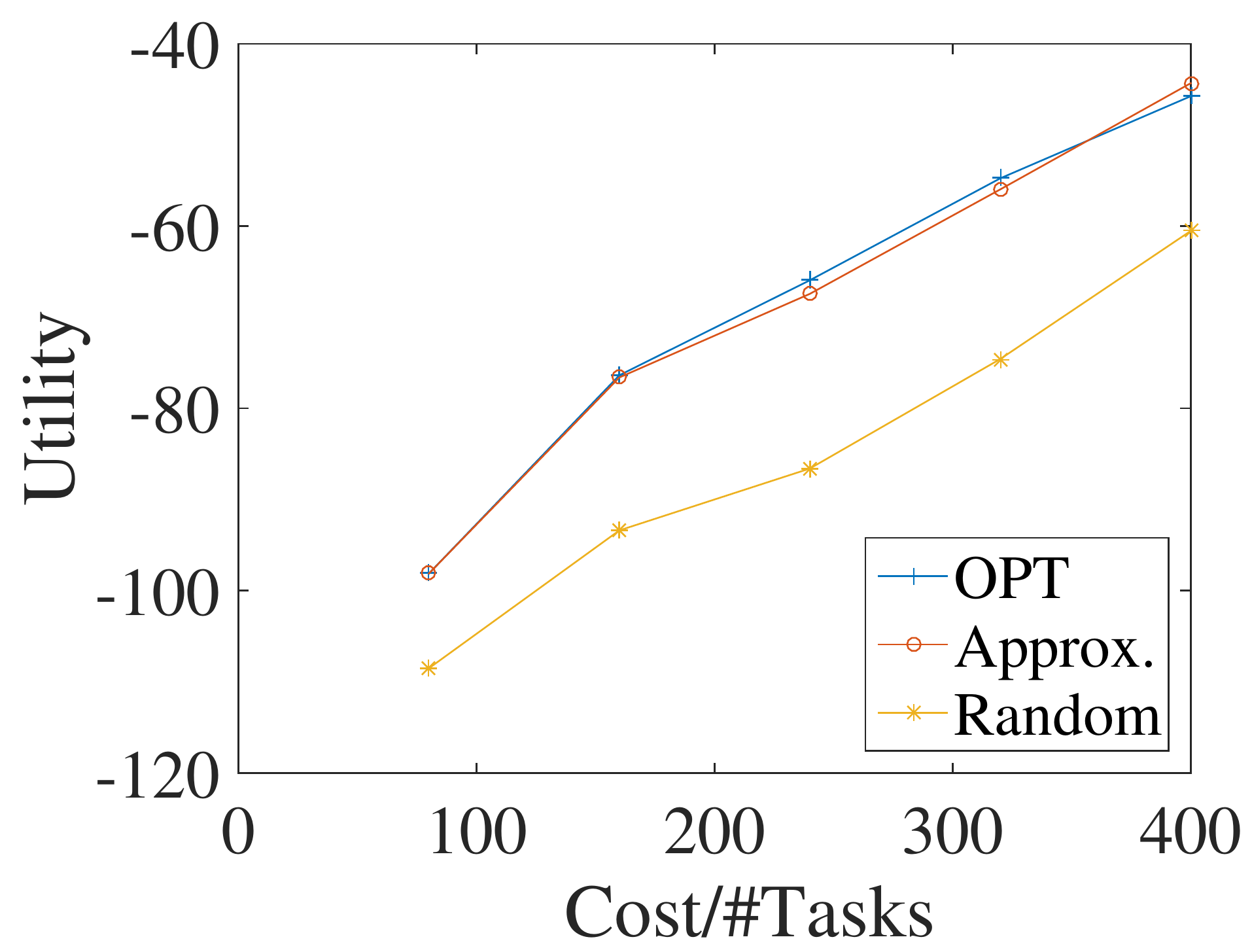}} 
	\subfigure[$P_c=0.9$, Utility]{ 
		\label{fig:bruteQuality:f} 
		\includegraphics[width=0.3\textwidth]{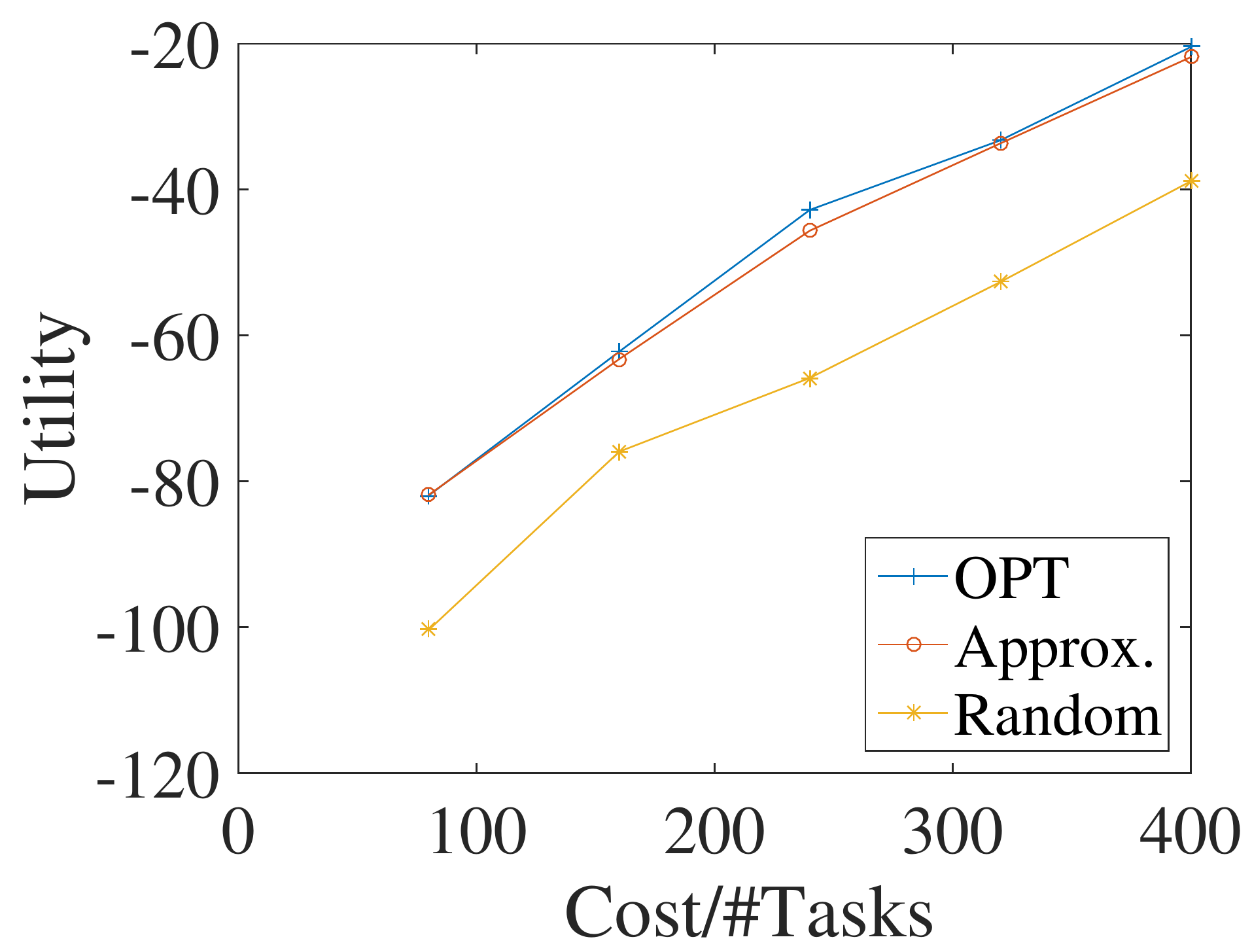}} 
	\centering
	\caption{Quality Improvement of Three Methods, Cost as Number of Total Tasks} 
	\label{fig:bruteQuality} 
\end{figure*}
\textbf{Dataset:} We adopt the \textit{Book} \footnote{\url{http://lunadong.com/fusionDataSets.htm}} dataset, which is widely used in the field of data fusion \cite{yin2008truth,dong2009integrating}. This dataset is particularly appropriate for crowd workers as it is about general real world knowledge and information that can be easily obtained manually.
We manually label the ground truth for all items related to the \textit{Gold standard} provided by the dataset for the purpose of evaluation.

As the data are statements of books' author list but not author of books, we define our facts triple as \{book, complete full name author list, statement\} rather than \{book, author, statement\}. Please note that since the author list could be in different formats and in different orders, for each book there could be more than one statement that is true. Here comes an example. For book \emph{Internet Effectively: A Beginner's Guide To The World Wide Web} with ISBN $0321304292$, the following two author lists statements are both considered as true statements: \emph{Adams, Tyrone; Scollard, Sharon} and \emph{Tyrone Adams, Sharon Scollard}.

Since our CrowdFusion system is general enough to be initialized by any ``machine-only" fusion model with probabilistic results,
we adopt the modified \textit{CRH framework}\cite{li2014resolving} for initialization in this experiment. Explicitly, statistics of a small set of books suggest that only around $50\%$ of Web data facts (i.e. raw data) is correct.
Because the \textit{CRH framework} only supports single true fact, we modify the \textit{CRH framework} with following method: we firstly mark top $50\%$ of author lists for each book as the correct author lists by majority voting and then we apply weight assignment, missing values normalization and truth computation in \textit{CRH framework} to the data set.


\begin{figure*} [!ht]
	\centering 
	\subfigure[$k=1\sim 3$, $F_1$-Scores]{ \label{fig:kIncreasing:a}
		\begin{minipage}{0.22\textwidth} 
			\centerline{$P_c=0.7$}
			
			\includegraphics[width=1\textwidth]{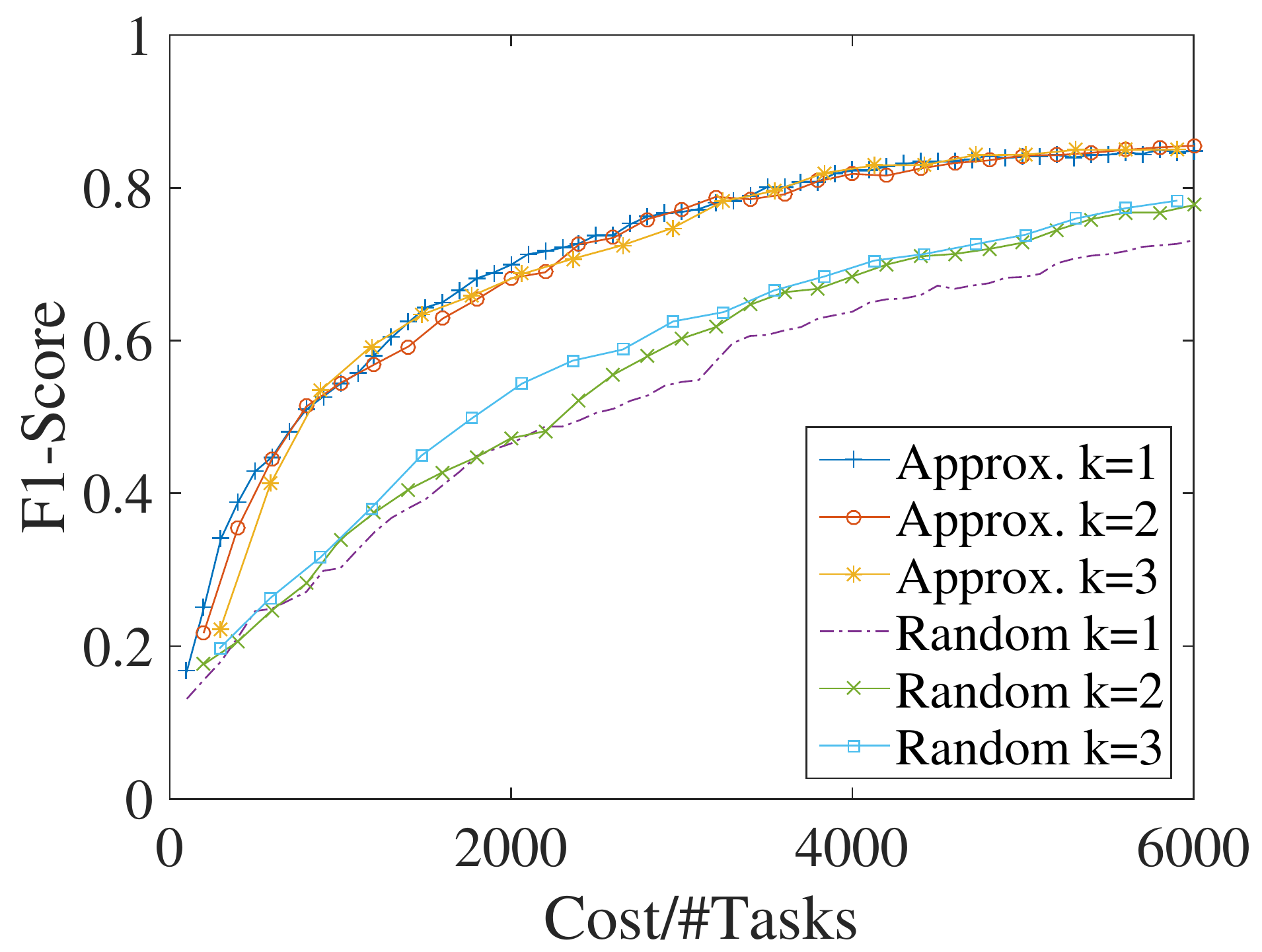} 
			\centerline{$P_c=0.8$}
			
			\includegraphics[width=1\textwidth]{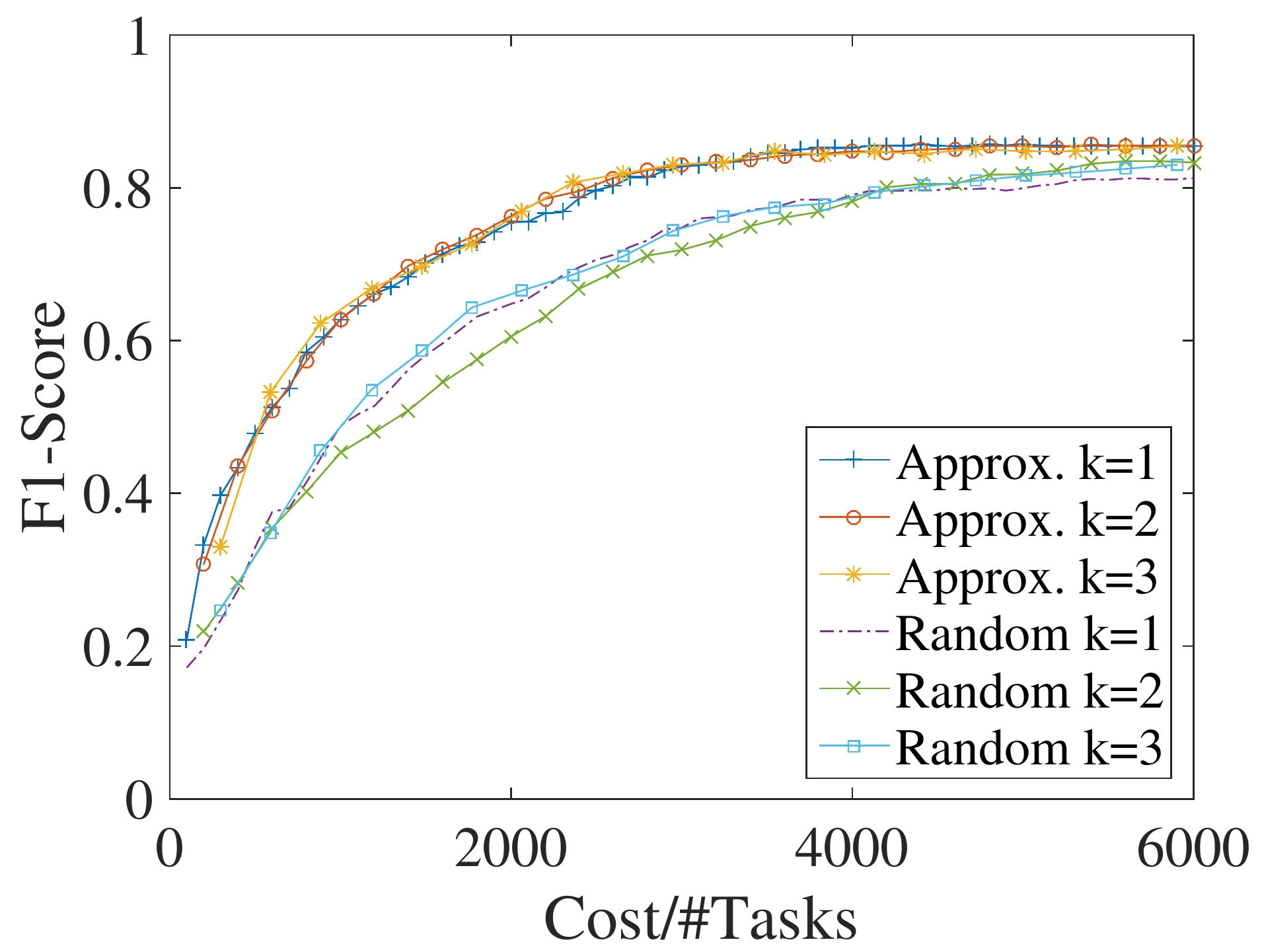} 
			\centerline{$P_c=0.9$}
			
			\includegraphics[width=1\textwidth]{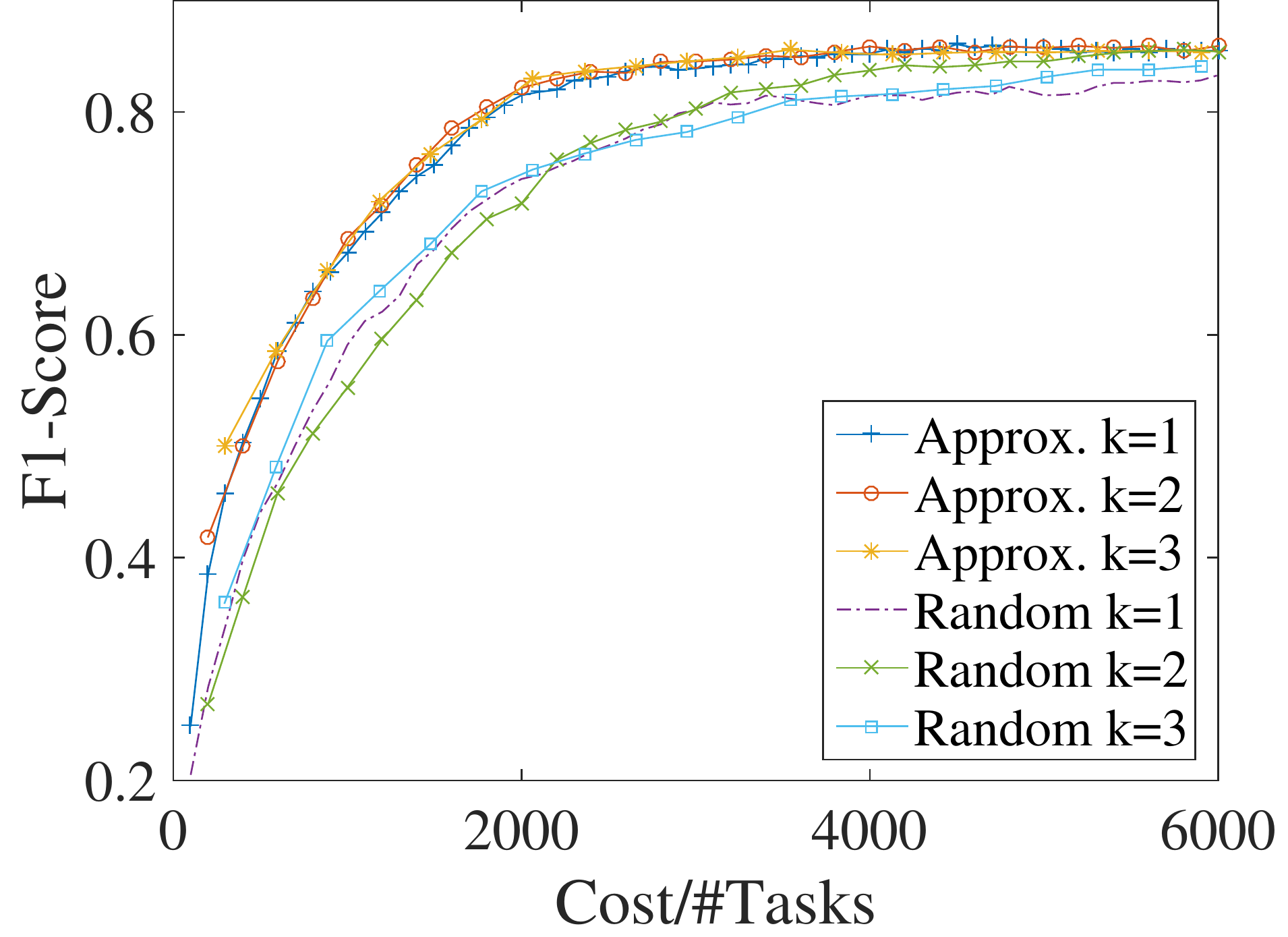} 
		\end{minipage} 
	} 
	\subfigure[$k=1\sim 3$, Utility]{ \label{fig:kIncreasing:b}
		\begin{minipage}{0.22\textwidth} 
			\centerline{$P_c=0.7$}
			
			\includegraphics[width=1\textwidth]{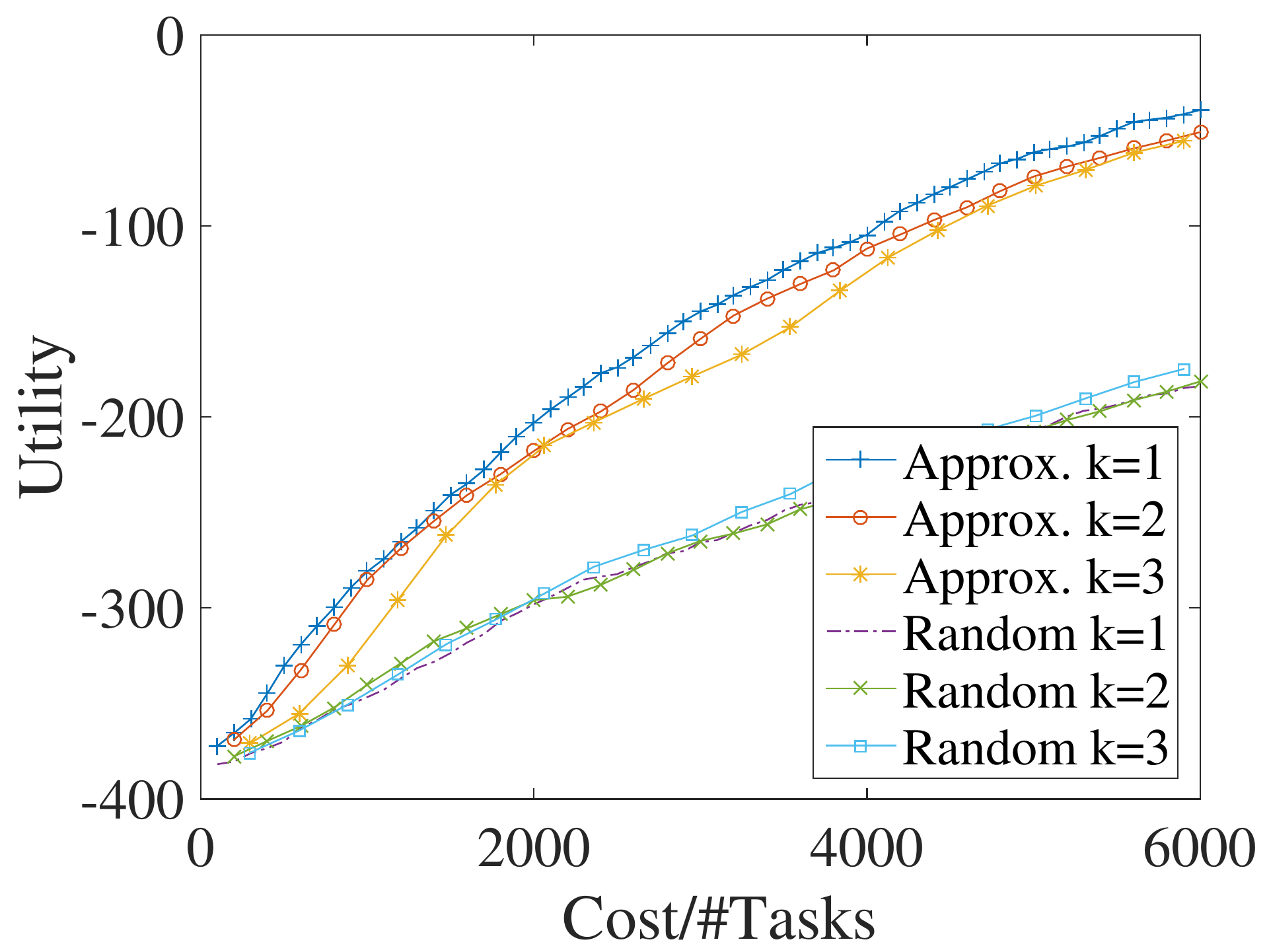} 
			\centerline{$P_c=0.8$}
			
			\includegraphics[width=1\textwidth]{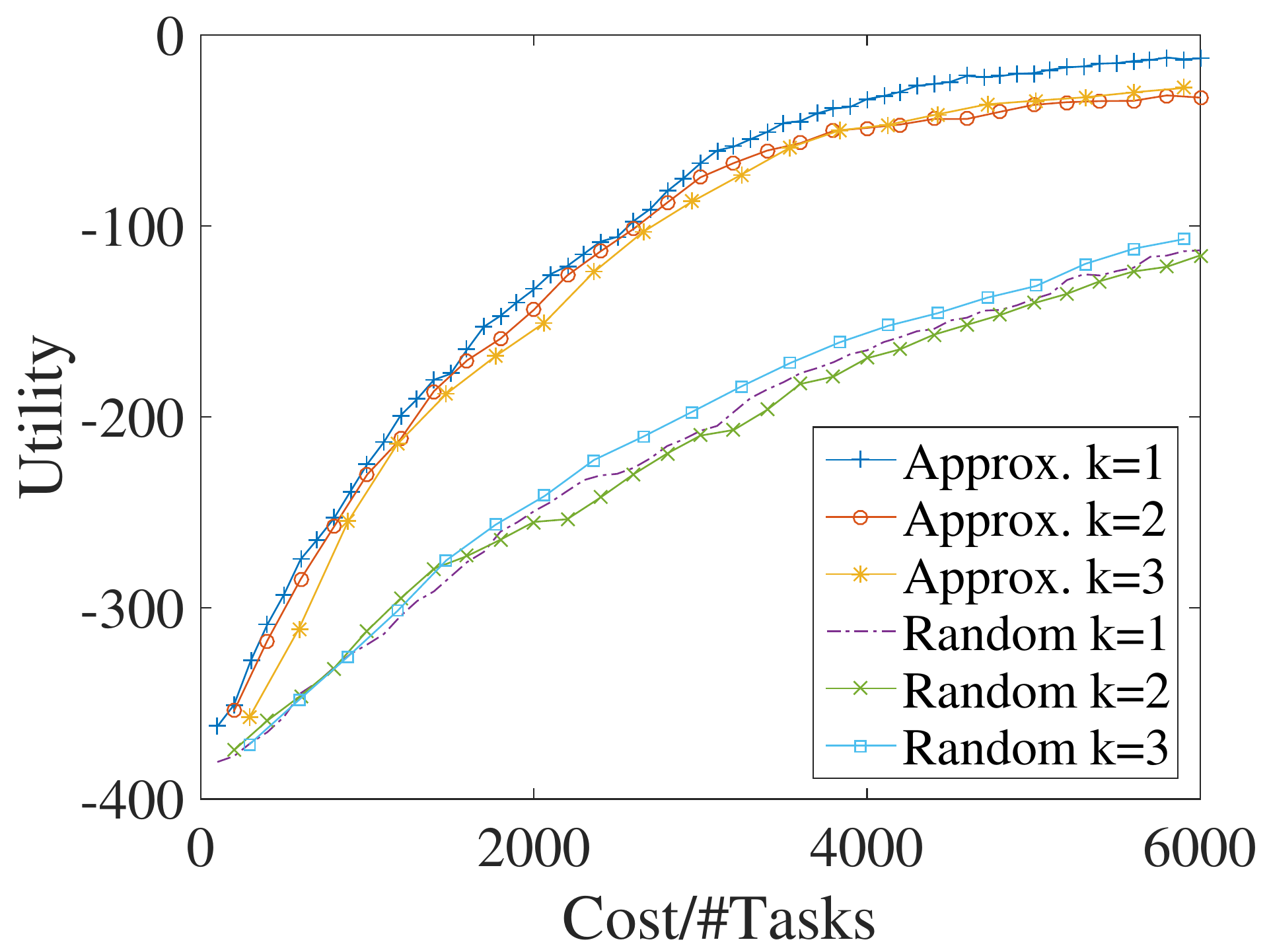} 
			\centerline{$P_c=0.9$}
			
			\includegraphics[width=1\textwidth]{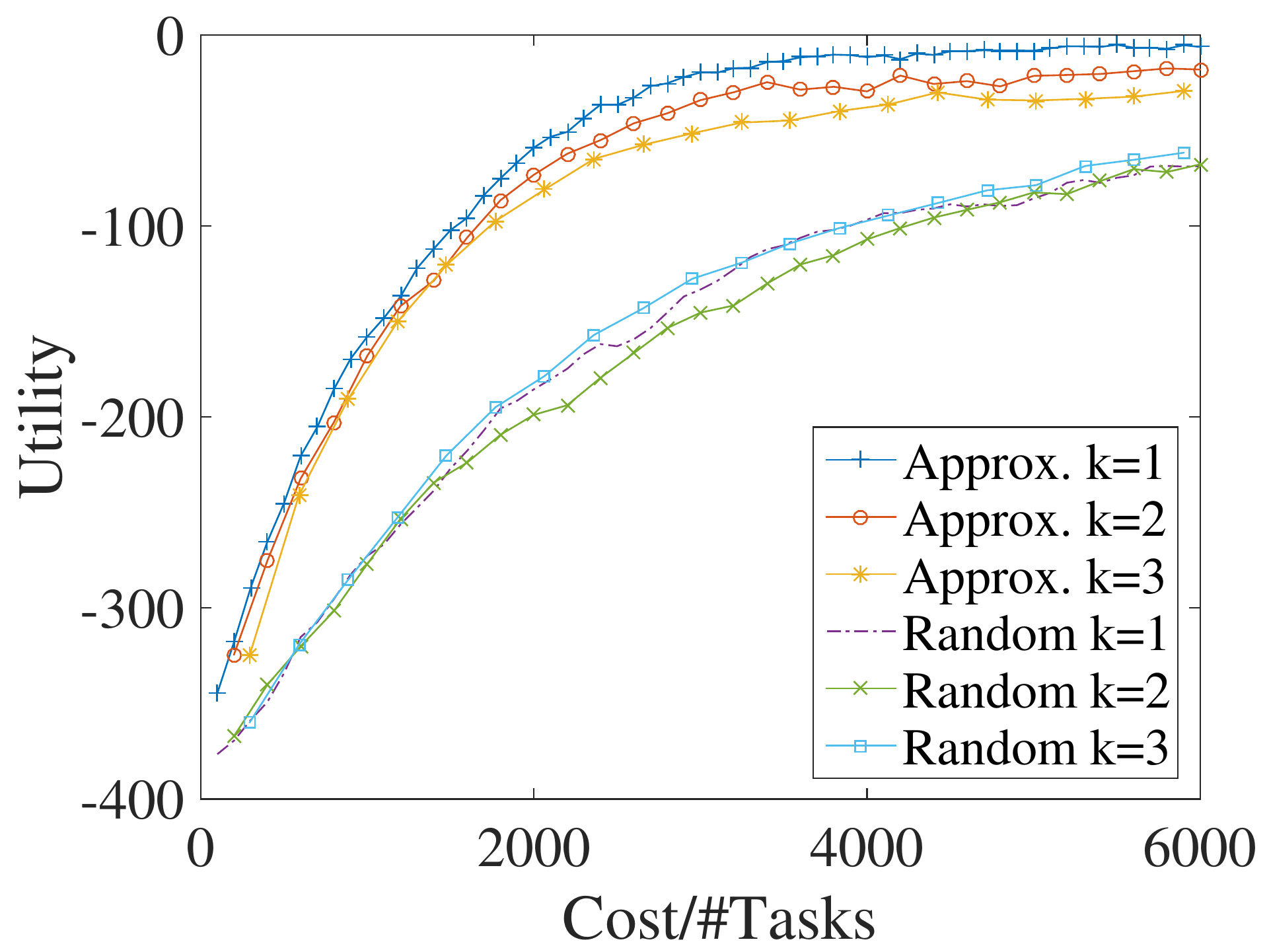} 
		\end{minipage} 
	} 
	\subfigure[$k=4\sim 6$, $F_1$-Scores]{ \label{fig:kIncreasing:c}
		\begin{minipage}{0.22\textwidth} 
			\centerline{$P_c=0.7$}
			
			\includegraphics[width=1\textwidth]{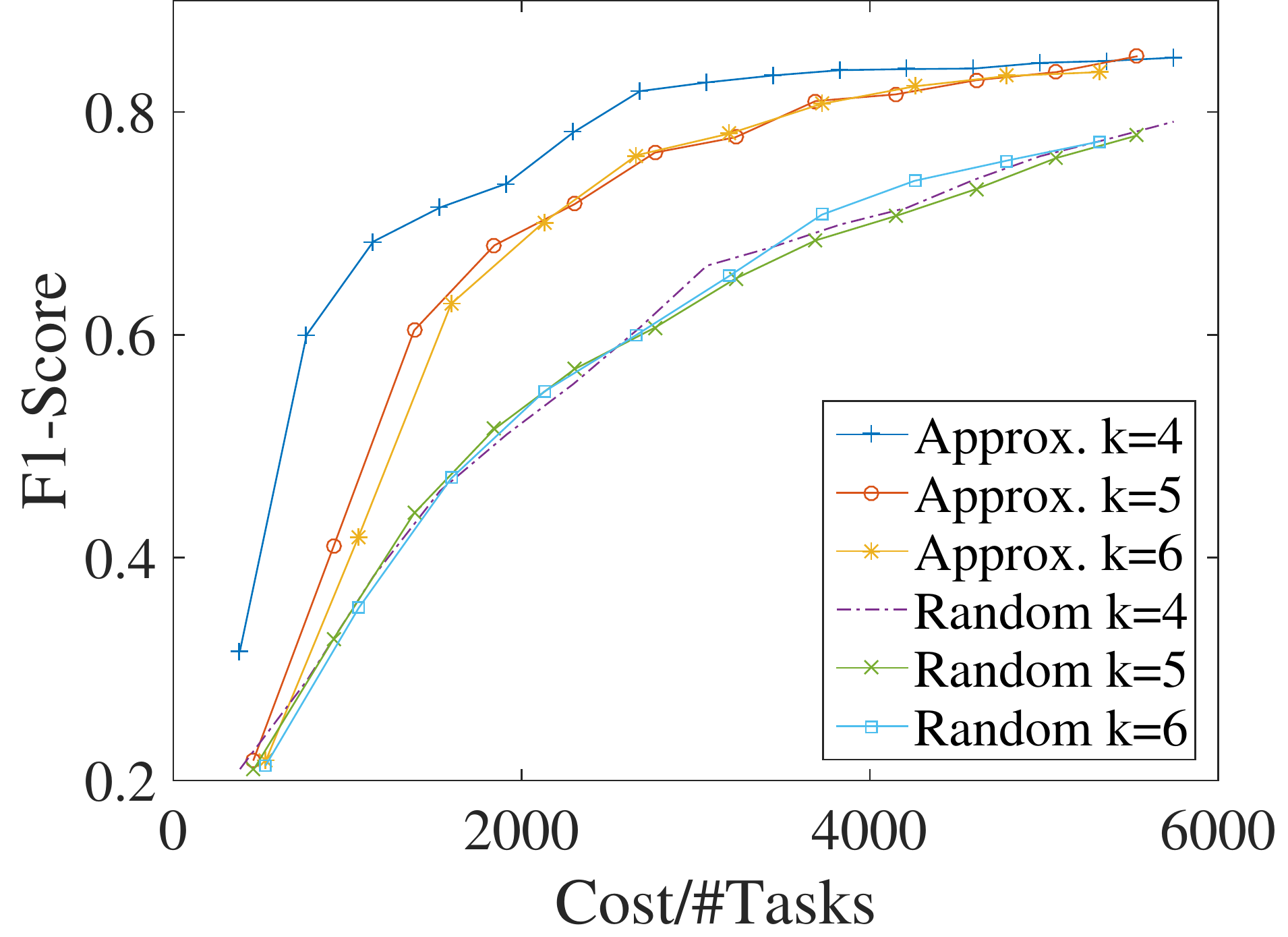} 
			\centerline{$P_c=0.8$}

			\includegraphics[width=1\textwidth]{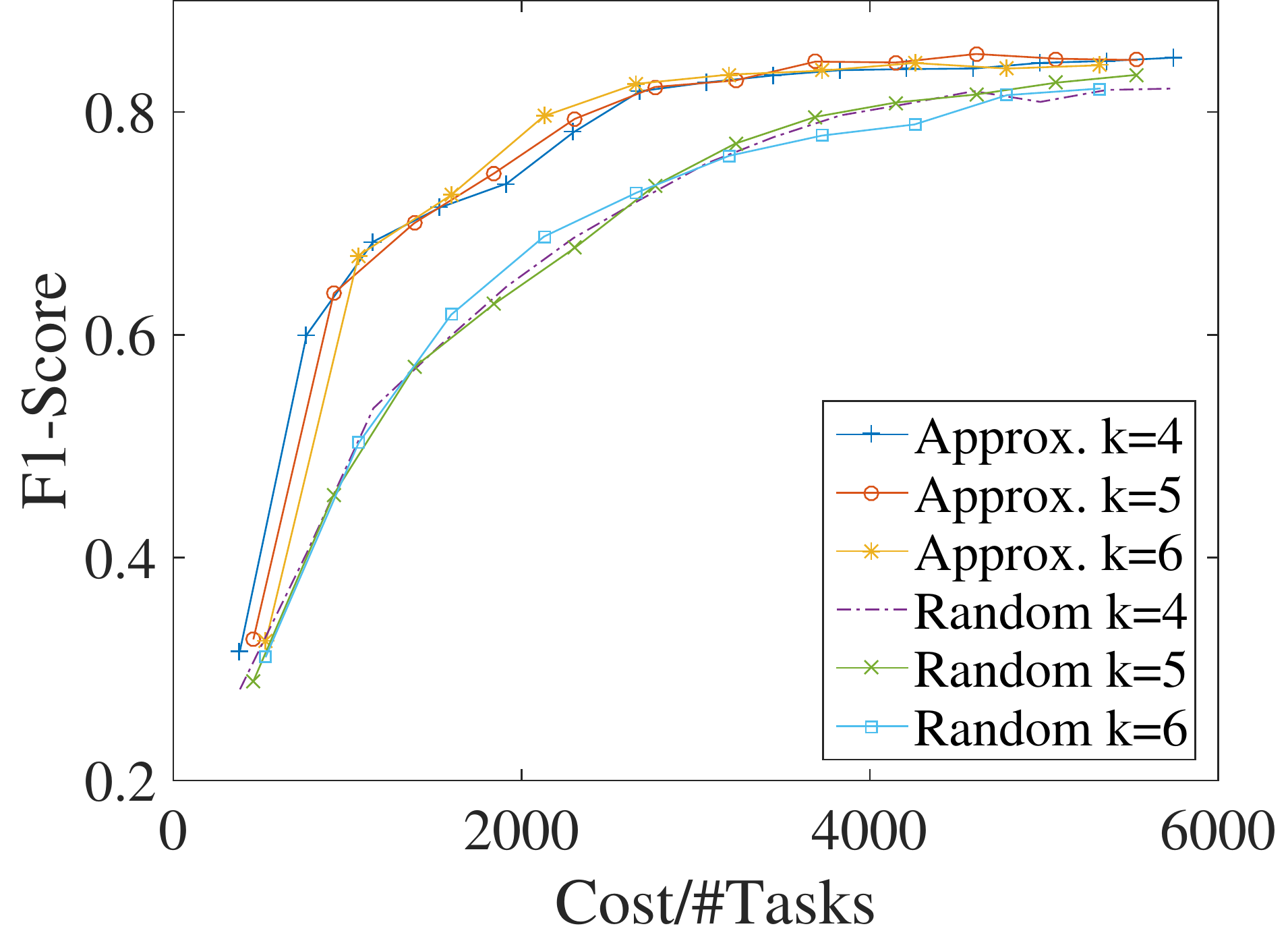} 
			\centerline{$P_c=0.9$}
			
			\includegraphics[width=1\textwidth]{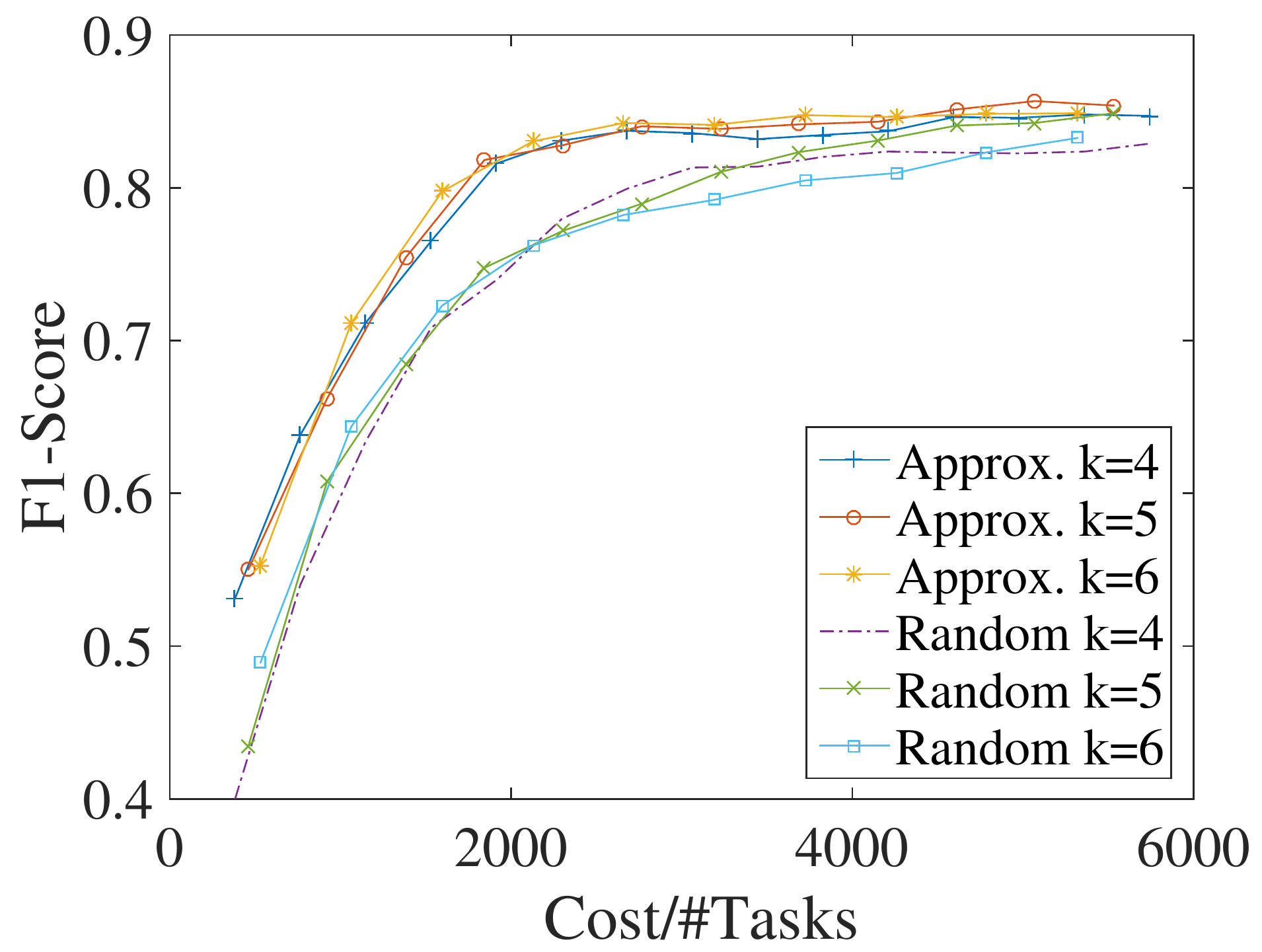} 
			
		\end{minipage} 
	} 
	\subfigure[$k=4\sim 6$, Utility]{ \label{fig:kIncreasing:d}
		\begin{minipage}{0.22\textwidth} 
			\centerline{$P_c=0.7$}
			
			\includegraphics[width=1\textwidth]{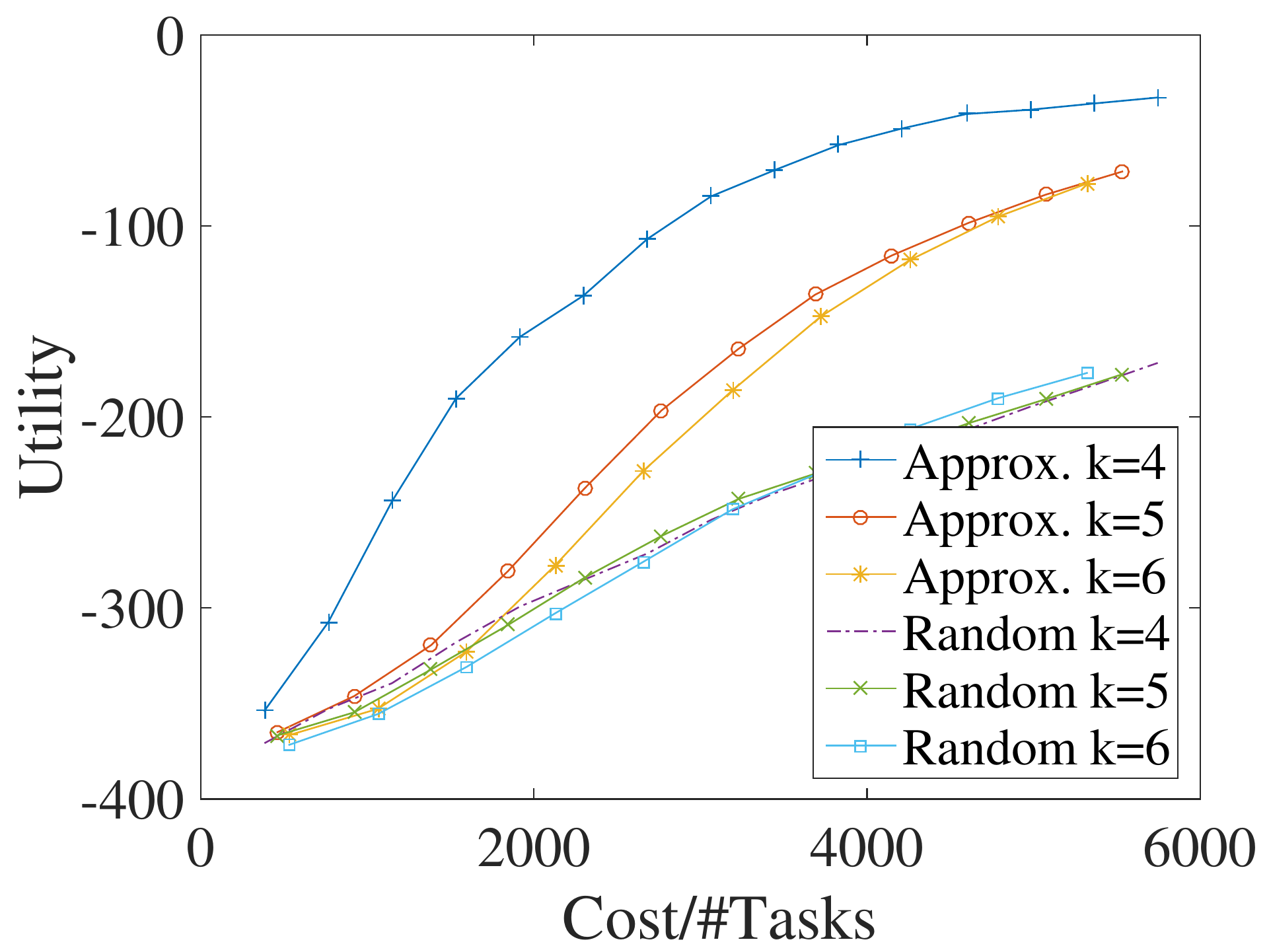} 
			\centerline{$P_c=0.8$}
			
			\includegraphics[width=1\textwidth]{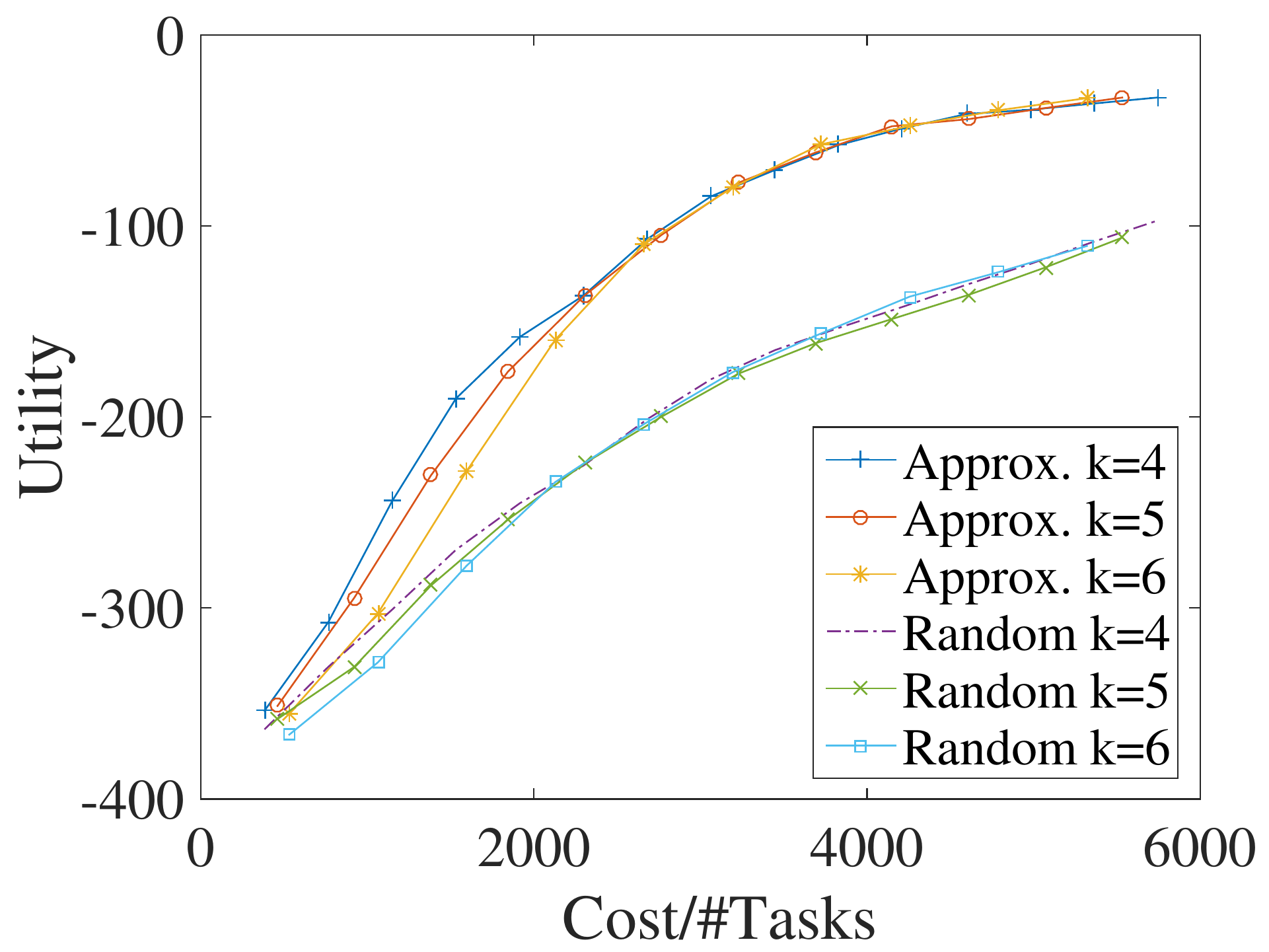} 
			\centerline{$P_c=0.9$}
			
			\includegraphics[width=1\textwidth]{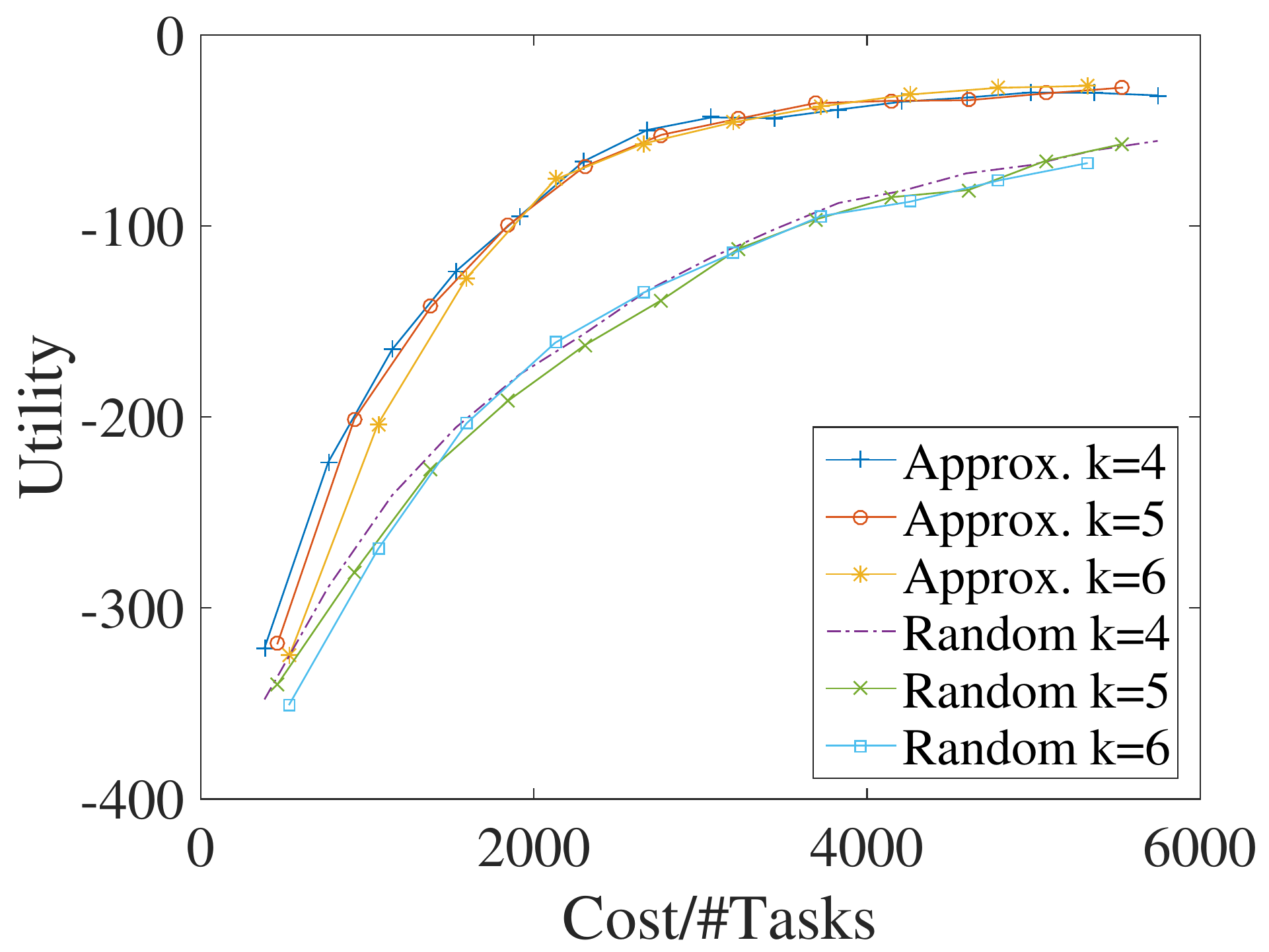} 
			
		\end{minipage} 
	} 
	\caption{Quality Improvement with Different $k$ settings, Cost as Number of Total Tasks} 
	\label{fig:kIncreasing} 
\end{figure*}
We treat information about each book independently and set a budget of $B=60$ tasks for each book.
In each round, we set the number of tasks to $k$, so there are $\lceil B/k \rceil$ rounds.
If a book has $n\geq k$ facts, we will ask $k$ tasks in every round except for the last one.
Otherwise, we will ask $n$ tasks in each round instead.

\textbf{Operating Environment}
The experiments are run on a $10$-nodes Linux cluster. Each note contains $4\times 10$-core Intel Xeon E$5$-$2650$v$3$ ($2.3$ GHz) processors and $196$ GB Physical memory. The Linux distribution installed is CentOS $6$, x$86\_64$ edition. For each condition, the programs are run for three times to get an average time cost on a single node of the cluster. 

\subsection{Efficiency Evaluation}\label{sect:time}
In this subsection, we show the time cost of the following competing algorithms: (1) \textbf{OPT}: selecting exact optimal algorithm by brute-force method; (2) \textbf{Approx.}: the approximation algorithm (3) \textbf{Approx.\&Prune} the approximation algorithm with pruning;  (4) \textbf{Approx.\&Pre.} the approximation algorithm with the preprocessing strategy and (5) \textbf{Approx.\&Prune\&Pre}: the approximation algorithm with both pruning and the preprocessing strategies. Please note that the books with a small number of facts usually stop getting better early and cannot show the efficiency clearly. Therefore, 
in order to distinguish the performances of above algorithms, we focus on books with facts more than $20$. We test average running time for the books in one round. The details of the experimental results are demonstrated in Table~\ref{table:time}.

The time cost of  \emph{OPT} method increases exponentially, which is not affordable in real application. 
With $k=4$, we had been waiting for more than $5$ days and the algorithm was still running. From the experimental results, one can see that the pruning strategy is powerful - the time cost is almost constant w.r.t. the increase of $k$, no matter we adopt the preprocessing or not.
Even though approximation method is a linear algorithm, the time cost increases rapidly and does not scale well. With the preprocessed data, we can significantly decrease the time consumption for task selection. Please note that we need to do the preprocessing for each round of selection. The approximation method with preprocessing is still a linear algorithm, but becomes much faster than that without preprocessing. 

\subsection{Quality Evaluation on gMission}
We verify the correctness of our CrowdFusion system, by
evaluating the accuracy of proposed algorithms with different settings of $k$ and $P_c$.

We use two different measurements to evaluate quality of our method. 
The first one is the utility defined in subsection \ref{sect:datamodel}, which is our optimization goal. This measurement indicates how well our algorithm approximates the optimization goal. We simply sum up the utility scores of all data instances for the evaluation. The second measurement is $F_1$ score, which is calculated based on the ground-truth labels. 

Budget is an essential issue in our experiment, thus we focus on the quality improvement with the increase of budget. Our budget allows CrowdFusion to ask at most $60$ tasks for each book and we have $100$ books in total.
\subsubsection{Comparison of Algorithms}
Due to that the \emph{OPT} solution is not scalable, we need to scale down $k$, $B$ and $n$ to conduct comparison with the \emph{OPT} solution. Please note that \emph{OPT} with $k=1$ would have the result exactly the same as that given by the approximation algorithms - they all select the very best task at each round. We compare performance of \emph{OPT} solution and our greedy solution only in condition that $k=2$ and budget $B=10$ with a small subset of data with $40$ books, which contains the least number of statements $n$ in the whole data set. 

Figure \ref{fig:bruteQuality} shows comparing between \emph{OPT} results, approximation algorithm results and a random selection results in both $F_1$-Score measurement and utility measurement. Our approximation strategy performs as good as the \emph{OPT} method not only in reaching optimization target but also in getting accurate results. And our approximation algorithm is significantly better than random selecting method. Additionally, \emph{OPT} results are not always better than the approximated results, which is because that the crowdsourced answer can be incorrect and hereby reduce the utility and F1-score. In other words, the quality is not absolute monotonic w.r.t the number of crowd sourced answers recieved. 
\subsubsection{$k$ Settings Comparison}
Recall that we will ask $\lceil B/k\rceil$ rounds for Data Fusion. With a limited budget $B$, smaller $k$ will cause more $rounds$.
Since we publish our tasks on gMission platform round by round, we need to wait for response from crowd workers one round by one round.
In each round, as every task will be distributed simultaneously, larger tasks set will not significantly increase waiting time for finishing answer collection.
Naturally, if we take larger $k$, we can finish our job faster. 

It is clear to see that the proposed algorithms outperform the random algorithm in any circumstance. In addition, from Figures \ref{fig:kIncreasing:b} and \ref{fig:kIncreasing:d}, one can observe that the curves
with smaller $k$ tend to have better performance in terms of utility. In Figure \ref{fig:kIncreasing:c} and \ref{fig:kIncreasing:d}, smaller $k$ also leads to higher $F_1$-Score. In particular, with $P_c = 0.7$, we get much higher $F_1$-Score if we set $k$ to a smaller value. This is because, when $k>1$, some tasks are not at the top
of the list when they are chosen, since at each iteration, we choose
$k$ good tasks rather than the very best one. Since there will be high uncertainty of workers answers if $P_c$ is low, we can ask questions more targetedly rather than trust user answers blindly. 
The experimental results suggest that the parameter $k$ can be used to balance the time efficiency and the quality. We conclude that
$k$ should be set to a small value when the budget is the main constraint; whereas a large value is suggested for k if time-efficiency
is the primary constraint. 

As another interesting observation, the random method shown in Figure \ref{fig:kIncreasing} indicates a totally reverse conclusion, i.e. the larger $k$ is, the better performance it gets.
In each round of task selection, we can only select a task once. Thus, large $k$ ensures that we can select a widely range of tasks in the random selection method. Rather than randomly selecting tasks multiple times, it would be better to select as wider range of tasks as possible. While rather than blindly randomly selecting, selecting tasks based on ultility estimation with our approximation method is the best.
Please note that if $k$ is bigger than $n$ for a subset, there would be no difference between the random selection and the approximated selection - they both select all the existing tasks.


\subsubsection{$P_c$ Settings Comparison}
\begin{figure}[t]\label{fig:PSelect}
	\centering 
	\subfigure[$P_c$ setting affects $F_1$-Score]{ \label{fig:PSelect:a}
		\includegraphics[width=0.40\textwidth]{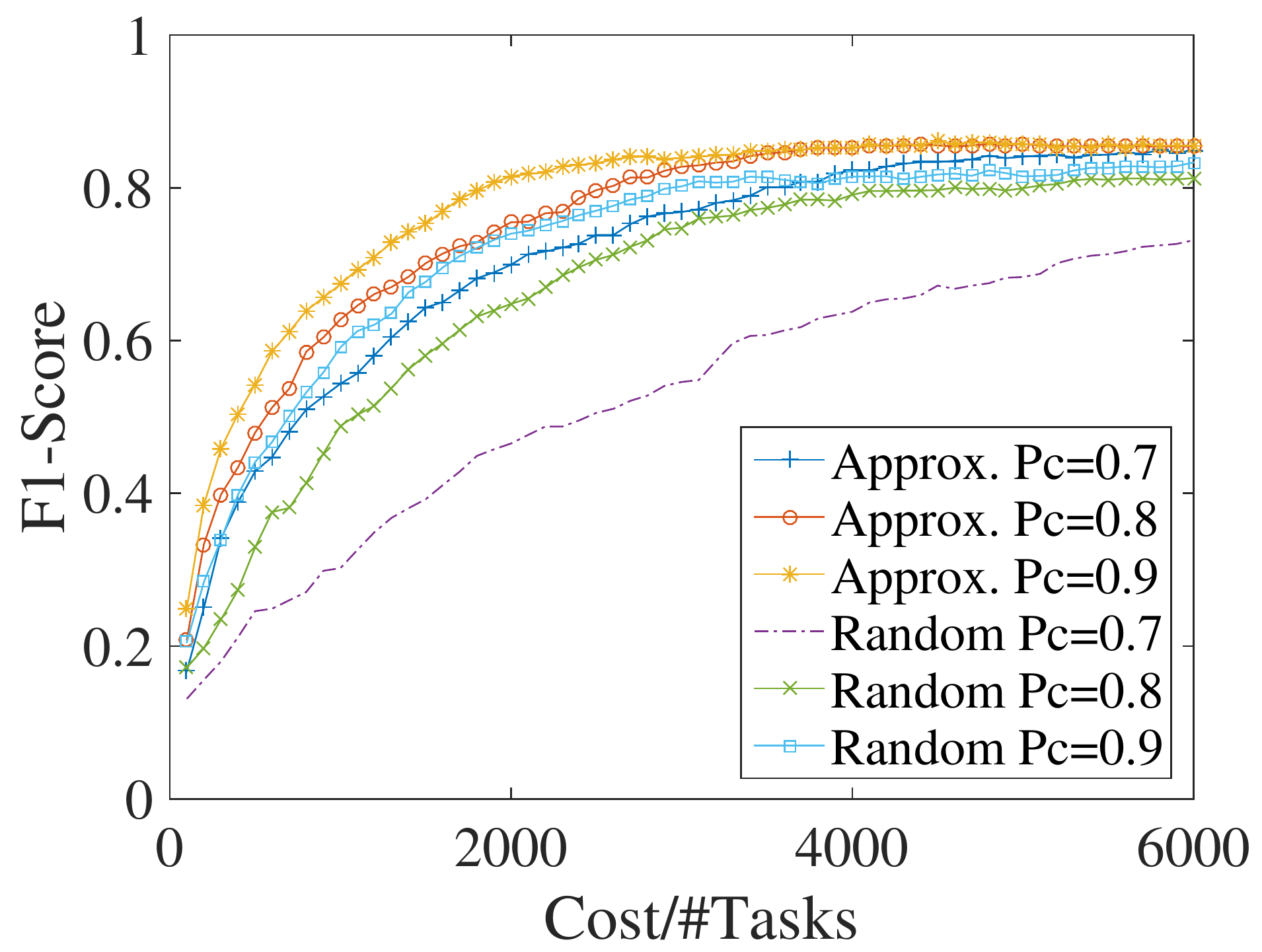} 
	} 
	\subfigure[$P_c$ setting affects Utility]{ \label{fig:PSelect:b}
		\includegraphics[width=0.40\textwidth]{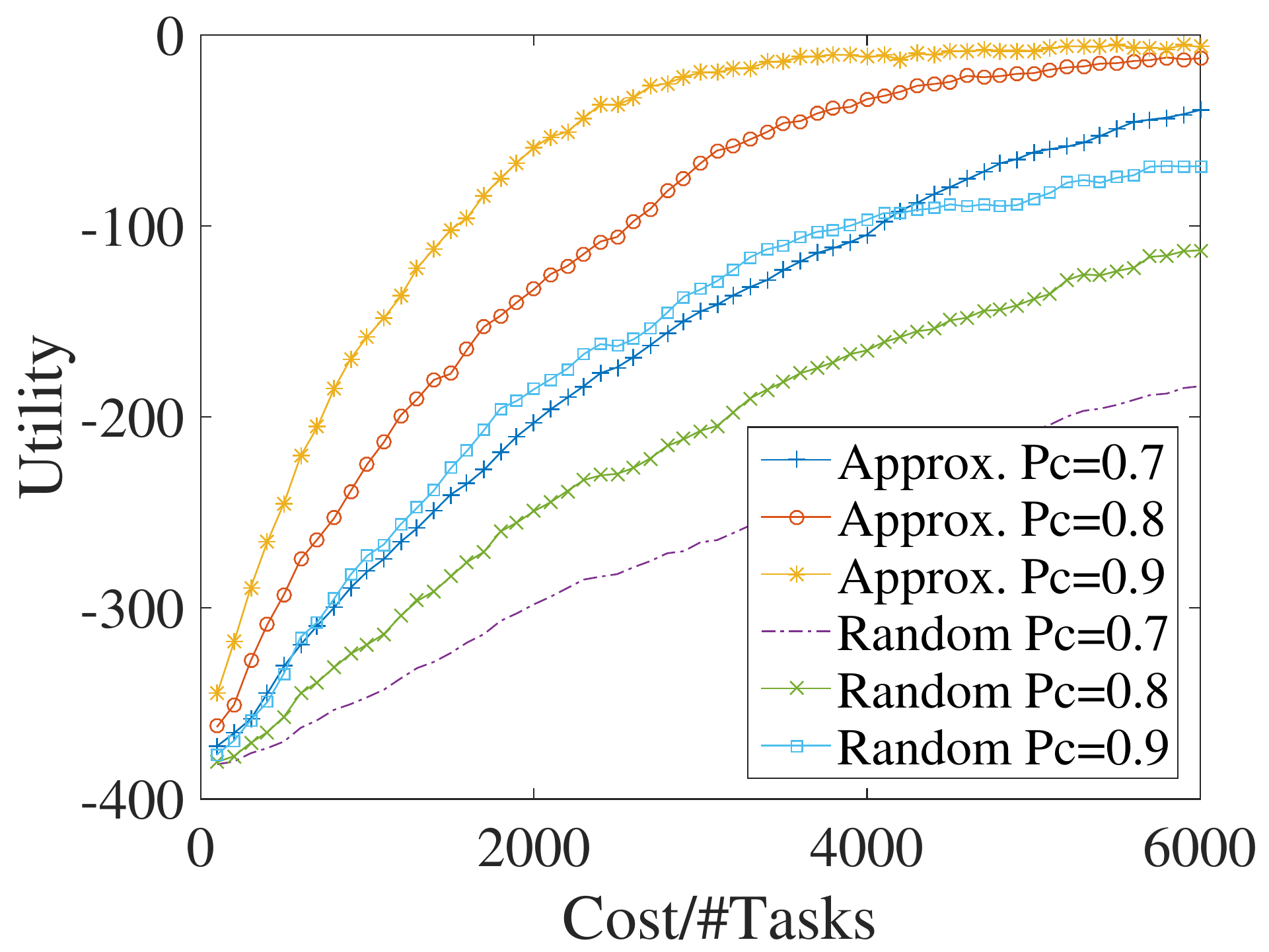} 
	} 
	\caption{Different $P_c$ settings, Cost as Number of Total Tasks}
\end{figure}
Consider an extreme condition first: if we set $P_c$ to $1$, we soon get the utility reduced to $0$, which is the maximum value for utility function.
That is to say, if crowd workers have high reliability, we will be able to confirm our answers by generating less tasks.
Similarly we can see from Figure \ref{fig:PSelect:b}, we can get higher utility with higher $P_c$.

However, with $P_c=1$, one error made by a worker will keep being wrong permanently. By statistical study, we know that the correct rate of workers for our task is around $0.86$. By setting $P_c$ to $0.8$ and $0.9$, both can achieve competitive $F_1$-score as we can see in Figure \ref{fig:PSelect:a}. Underestimating the reliability of the crowd would slow down the overall crowdsourcing procedure and proper estimating the reliability can lead to a good performance. Consequently, if possible, in real applications, we should estimate the reliability by a pre-test with groundtruth.

Utility value and $F_1$-Score are different measurements. Utility is for measuring how confident we are about our judgment and $F_1$-Score is for measuring accuracy compared with ground truth considering both precision and recall. By all comparisons till now, we see high consistency between our utility function value and $F_1$-score. This fact indicates that our utility function can properly fit into real-world data fusion applications very well.

\subsection{Error analysis}
No matter how we choose parameters, there still remains a gap between the obtained $F_1$-Score and $1$. In this section we will discuss about what causes the gap and how to fit the gap.

By manually checking, we conclude that there are two kinds of wrong judgments.

First, books with large numbers of statements are more likely to be judged incorrectly.
As we have budget $B=60$ for each book, the books with more statements have less opportunities to be asked for each statement.
The difference between smaller $k$ and larger $k$ is mostly caused by this part of errors. 
With smaller $k$, we can target on facts which can increase data quality most significantly.
However, with a large $k$, we can neither ask broad range of statements for those books nor ask crowds with targets.
This can be solved by using smaller $k$ or increasing budget.
Also, if a proper strategy can be designed to distribute budgets among all subsets of facts, this can be solved.

Second, there often exist statements that can hardly be judged.
We notice that crowd answers to certain statements only have correct rate a little bit higher than $0.5$ and our system will give judgment based on correlations between the facts and other facts consequently.
The confusion of judgment appears on following types of statements:
\begin{itemize}
	\item[*] \textbf{Wrong Order} In our ground truth, different author list order will not affect the judgment of whether the author list is correct or not. The book with ISBN $1558609350$ has two authors \emph{Catherine Courage and Kathy Baxter} as shown on cover page. There is a correct statement that the author list is \emph{BAXTER, KATHY COURAGE, CATHERINE}. However, such an author list statement will lead to high diversity of crowd answers. This is the most significant error judgment and has caused a lot of false negatives;
	\item[*] \textbf{Additional Information} In some statement, organization of authors or publisher is included. Such as a statement about author of book with ISBN $0201767910$ is \emph{RUCKER, RUDY (SAN JOSE STATE UNIVERSITY, USA)}. Our ground truth considers it as wrong statement since organization information is not a part of ``author list''. However we notice that more than $40\%$ of workers consider such a statement as true.
	\item[*] \textbf{Misspelling} A misspelling in an author list is hard to be noticed and for some statement the correct rate is even lower than $50\%$ only because of the misspelling. There is a book with ISBN $1558608109$ authored by \emph{Pete Loshin}. However, more than a half crowd workers mark statement \emph{Loshin, Peter} correct.
\end{itemize}

Such kinds of error are caused by lack of judgment standard and ambiguous guidance. And this can be solved by giving workers a small set of sample tasks and offering them correct answers with explanation of that answers. As long as we ensure that correct rate is higher than $50\%$, the errors can also be solved by increasing budget.

\section{Related work}
\label{related}
In this section, we review the related work. Since our system is built on existing data fusion methods and crowdsourcing technique is applied, we review the related work in two categories: (1) data management, query processing and learning with crowds and (2) data fusion and truth discovery.

\subsection{Data Management, Query Processing and Learning with Crowds}
Many tasks cannot be answered by machines only. The recent booming up of crowdsourcing brings us a new opportunity to engage human intelligence into the process of answering such queries (see \cite{DBLP:journals/cacm/DoanRH11} as a survey).  In general, \cite{DBLP:conf/sigmod/FranklinKKRX11} proposed a query processing system using microtask-based crowdsourcing to answer queries. Many classical queries are studied in the context of crowdsourced database, including max \cite{guo2012so}, filtering \cite{parameswaran2012crowdscreen}, sorting \cite{journals/corr/abs-1109-6881} etc. In \cite{DBLP:conf/cidr/ParameswaranP11}, a declarative query model is proposed to cooperate with standard relational database operators. As a typical application which relates to data integration, \cite{wang2012crowder} utilized a hybrid human-machine approach on the problem of entity resolution. We follow the general way which utilizes a hybrid system for data fusion problem.

On the other hand, related researches have been studying active learning via crowds. Active learning is a form of supervised machine learning, in which a learning algorithm is able to interact with the experts (or some other information source) to obtain the desired outputs at new data points. Thus, the goal of active learning is to improve the accuracy of classifiers as much as possible through selecting limited data to label. A widely used technical report on active learning is \cite{DBLP:series/synthesis/2012Settles}. In particular, \cite{DBLP:journals/corr/abs-1209-3686,DBLP:conf/aaai/ZhaoSS11} proposed active learning methods specially designed for crowd-sourced databases. Our method follows the general principle of active learning. When there is a need, we sample data from crowd with an estimation of the information gain.

\subsection{Data Fusion and Truth Discovery}

Data fusion is one of the approaches to integrate multiple data
sources, which is widely applied in tasks of discovering truth.
Differences between different data fusion strategies are about how
to solve conflicts \cite{bleiholder2008data}.

The strategies include conflict ignoring, conflict avoiding and
conflict resolving\cite{bleiholder2006conflict}. Recent researches
move attention to Web data fusion. Due to dynamic millions of web
data with widely differing qualities, new strategies are developed
to deal with  the problem of lacking of cleaning evidence
\cite{rezig2015query}, estimating web sources reliability
\cite{dong2015knowledge} and information propagating detection
\cite{bronselaer2015propagation}.

As surveyed in \cite{li2015survey}, data fusion can be categorised by scoring or labeling. The scoring usually in the form of probability which is suitable for our method for an improvement. A labeling result assign a true or false label to each of facts. Such a result can turn into scoring and be further improved by our system by initializing with probability measurement of how trustworthy the machine-only method is.

\section{Conclusions}\label{sect:conc}
In this paper, a crowdsourced data fusion refinement method is proposed and utilized to improve data fusion results of existing machine-only methods.
In lots of circumstances, judging on-line information is true or false is a complicated task for machines as information online is neither structured nor able to comprehended by machines.
Consequently, we make use of human intelligence to help us with high quality data fusion.
With the consideration that crowd is noisy, our result can benefit from the crowd answers to the selected tasks.
Since different tasks may lead to different amount of benefit, we design an approximate algorithm with pruning and preprocessing strategies for this task selection problem which is NP-hard if we want exact the best task set.
Empirical study shows that CrowdFusion achieves high accuracy at
finding true facts and at the same time computational cost can be reduced by
the approximation algorithm or heuristic solution without losing
much effectiveness.

Many existing data fusion methods can be applied to CrowdFusion by
considering their result confidence distribution as an input of
CrowdFusion system. To be specifically, Bayesian based data fusion
methods and probability based data fusion methods can be simply
applied to CrowdFusion because their result is a (marginal)
probability distribution and can be extended to the joint
distribution as required. With the help of crowds, we can obtain
high quality data fusion result even if lacking of domain specific
knowledge. Such an easy-to-apply method can greatly benefit existing
data fusion methods without modification of any existing method.

Our work is just an initial solution for data fusion with
crowdsourcing technique As there is only one domain specific
heuristic method to model inference relationship between facts
\cite{yin2011semi} and correlation between facts can improve
performance of CrowdFusion, further research on measuring the
relationships between facts is a direction, which may be related
with natural language processing, image processing and audio
processing.


%

\bibliographystyle{abbrv}
\bibliography{crowdfusion}

\end{document}